\documentclass[11pt,a4paper]{article}
\usepackage{bbding}
\fussy
\usepackage[pagebackref=false,citecolor=newblue,colorlinks=true,linkcolor=newblue,bookmarks=false]{hyperref}

\usepackage[font=scriptsize,bf]{caption}
\usepackage{epsfig,amssymb,amsfonts,amsmath,amsthm}
\usepackage[margin=1in]{geometry}

\usepackage{algorithm}
\usepackage{algorithmic}
\usepackage{amsopn}

\usepackage{xcolor}

\usepackage{rotating}
\usepackage{minitoc}
\usepackage{enumerate}
\definecolor{ocre}{rgb}{0.72,0,0} 
\definecolor{newblue}{rgb}{0.2,0.2,0.6} 

\definecolor{ocre}{rgb}{0.72,0,0} 

\definecolor{babyblueeyes}{rgb}{0.63, 0.79, 0.95}

\definecolor{newgreen}{rgb}{0.53,0.66,0.42} 

\newcommand{\remove}[1]{}

\newcommand{\N}{\mathbb{N}}
\newcommand{\R}{\mathbb{R}}

\newcommand{\rot}{\intercal}
\newcommand{\ce}{\mathrm{e}}

\newcommand{\calR}{\mathcal{R}}

\newcommand{\poly}{\operatorname{poly}}

\newcommand{\APT}{\mathsf{APT}}
\newcommand{\eps}{\epsilon}
\newcommand{\calX}{\mathcal{X}}

\newcommand{\core}{\mathsf{CORE}}

\newcommand{\vol}{\operatorname{vol}}

\newcommand{\COST}{\mathsf{COST}}

\newcommand{\argmax}{\operatorname{argmax}}

\DeclareMathOperator{\spn}{span}
\DeclareMathOperator{\dmn}{dim}
\renewcommand{\leq}{\leqslant}
\renewcommand{\geq}{\geqslant}
\renewcommand{\le}{\leqslant}
\renewcommand{\ge}{\geqslant}

\newcommand{\thmref}[1]{Theorem~\ref{thm:#1}}

\newcommand{\lemref}[1]{Lemma~\ref{lem:#1}}

\newcommand{\secref}[1]{Section~\ref{sec:#1}}

\newcommand{\eq}[1]{\eqref{eq:#1}}

\renewcommand{\tilde}{\widetilde}
\renewcommand{\epsilon}{\varepsilon}

\newcommand{\aseq}{\{A_i\}_{i=1}^k}
\newcommand{\sseq}{\{S_i\}_{i=1}^k}
\newcommand{\AlgoMean}{\mathsf{AlgoMean}}

\definecolor{ocre}{RGB}{150,22,11} 

\newcommand{\calL}{\mathcal{L}}

\usepackage{setspace}
\def\argmax{\operatornamewithlimits{argmax}}

\newtheorem{thm}{Theorem}[section]  

\newtheorem{lem}[thm]{Lemma}

\newtheorem{cor}[thm]{Corollary}

\newtheorem{prob}{Problem}
\renewcommand{\tilde}{\widetilde}
\numberwithin{equation}{section}
\usepackage{todonotes}

\newcommand{\mat}[1]{\boldsymbol{\mathbf{#1}}}

\newcommand{\citep}{\cite}

\usepackage{cleveref}
\crefname{thm}{Theorem}{Theorems}

\title{\textbf{Partitioning Well-Clustered Graphs: \\
Spectral Clustering Works!}\thanks{A preliminary version of this paper appeared in the 28th Annual Conference on Learning Theory~(COLT 2015).}}

 \author{Richard Peng
 \thanks{Georgia Institute of Technology, Atlanta, USA.
 (\url{rpeng@cc.gatech.edu})}
 \and
He Sun\thanks{University of Bristol, Bristol, UK.
 (\url{h.sun@bristol.ac.uk}) Questions, comments, or corrections
to this document may be directed to that email address.}
\and
Luca Zanetti\thanks{University of Bristol, Bristol, UK. (\url{luca.zanetti@bristol.ac.uk})
}
}
\date{}

\begin{document}

\maketitle

\begin{abstract}
In this paper we study variants of the widely used
\emph{spectral clustering}  that
 partitions a graph into $k$ clusters by (1) embedding the vertices of a graph into a low-dimensional space using the bottom eigenvectors of the Laplacian matrix, and (2) grouping the embedded points into $k$ clusters via $k$-means algorithms. We show that, for a wide class of  graphs,
spectral clustering  gives a good approximation of the optimal clustering.
While this approach was proposed in the early 1990s and has comprehensive applications, prior to our work  similar results were known only for graphs generated from stochastic models.

We also give a nearly-linear time algorithm for partitioning well-clustered graphs based on   computing a matrix exponential and
approximate nearest neighbor data structures.

\vspace{0.5cm}


\end{abstract}

\thispagestyle{plain}

\thispagestyle{empty}

\setcounter{page}{0}

\newpage

\thispagestyle{empty}

\tableofcontents

\setcounter{page}{0}

\newpage

\section{Introduction}

Partitioning a graph into two or more  pieces
is one of the most fundamental problems in combinatorial optimization,
and has comprehensive applications in various disciplines of computer science.

One of the most studied graph partitioning problems is the \emph{edge expansion problem}, i.e., finding a cut with few crossing edges normalized by  the size of the smaller side of the cut.
Formally, let $G=(V,E)$ be an undirected  graph. For any set $S$,  the conductance of set $S$ is defined by
\[
\phi_G(S)\triangleq \frac{ |E(S, V\setminus S)|}{ \vol(S)},
\]
where $\vol(S)$ is the total weight of edges incident to vertices in $S$,
and let the conductance of $G$ be
\[
\phi(G)\triangleq \min_{S:\vol(S)\leq \vol(G)/2} \phi_G(S).
\]
The edge expansion problem asks for a set $S\subseteq V$ of $\vol(S)\leq \vol(V)/2$ such that $\phi_G(S)=\phi(G)$.
This problem is known to be $\mathsf{NP}$-hard~\cite{MatulaS90}, and the current best approximation algorithm achieves an approximation ratio of $O\left(\sqrt{\log n} \right)$~\cite{ARV09}.

The \emph{$k$-way partitioning problem} is a natural  generalization of the edge expansion problem.
We call subsets of vertices~(i.e.~\emph{clusters}) $A_1,\ldots, A_k$ a \emph{$k$-way partition} of $G$ if $A_i\cap A_j=\emptyset$ for different $i$ and $j$, and $\bigcup_{i=1}^k A_i=V$. The $k$-way partitioning problem asks for a $k$-way partition of $G$  such that the conductance of any
$A_i$ in the partition is at most  the \emph{$k$-way expansion constant}, defined by
\begin{equation}\label{eq:defrho}
\rho(k)\triangleq\min\limits_{\mathrm{partition\ } A_1,
\ldots, A_k }\max_{1\leq i\leq k} \phi_G(A_i).
\end{equation}
Clusters of low conductance in networks appearing in practice usually capture the notion of \emph{community},
and algorithms for finding these subsets have applications in various domains such as
community detection and network analysis.
In computer vision,  most image segmentation procedures are
based on region-based merge and split~\citep{CC79}, which in turn
rely on partitioning graphs into multiple subsets~\citep{ShiM00}.
On a theoretical side, decomposing vertex/edge sets into multiple disjoint subsets
is used in designing approximation algorithms for Unique Games~\citep{toc/Trevisan08},
and efficient algorithms for graph problems~\citep{KelnerLOS14,LR99,spielmanTeng:SS11}.

Despite widespread use of various graph partitioning schemes over the past decades, the quantitative relationship between the $k$-way expansion constant and the eigenvalues of the graph Laplacians were unknown until a sequence of very recent results~\citep{conf/stoc/LeeGT12,lrtv12}. For instance, 
Lee et al.~\cite{conf/stoc/LeeGT12} proved the following higher-order Cheeger inequality:
\begin{equation}\label{eq:highorder}
\frac{\lambda_k}{2}\leq\rho(k)\leq O(k^2)\sqrt{\lambda_k},
\end{equation}
where $0=\lambda_1\leq \ldots\leq\lambda_n\leq 2$ are the eigevalues of the normalized Laplacian matrix $\calL$ of $G$.
Informally, the higher-order Cheeger inequality shows that a graph $G$ has a $k$-way partition with low $\rho(k)$ if and only if $\lambda_k$ is small.
Indeed, \eq{highorder} implies that  a large gap between $\lambda_{k+1}$ and $\rho(k)$ \emph{guarantees} (i) existence of a $k$-way partition  $\{S_i\}_{i=1}^k$ with bounded $\phi_G(S_i)\leq \rho(k)$, and (ii)  any $(k+1)$-way partition of $G$   contains a subset with significantly higher conductance $\rho(k+1)\geq \lambda_{k+1}/2$ compared with $\rho(k)$. Hence, a suitable lower bound on the \emph{gap} $\Upsilon(k)$ for some $k$, defined by
\begin{equation}\label{eq:defupsilon}
\Upsilon(k)\triangleq\frac{\lambda_{k+1}}{\rho(k)},
\end{equation}
implies the  existence of a $k$-way partition  for which
every cluster has low conductance, and that $G$ is a \emph{well-clustered} graph.

We study  well-clustered graphs which satisfy  a
gap assumption on $\Upsilon(k)$ in this paper. Our gap assumption on $\Upsilon(k)$  is slightly weaker than assuming gaps between the eigenvalues, but nonetheless related via Cheeger-type inequalities. Our assumption is also well-grounded in practical studies:
clustering algorithms have  been studied before under
this assumption in machine learning, e.g.~\cite{conf/icml/ZhuLM13}.
Sharp drop-offs between two consecutive eigenvalues have also been
observed to give good indicators for the number of clusters,
e.g.~\cite{luxburg07} and Section D in~\cite{fortunatoPR}.

\subsection{Our Results}

We give structural results that  show close connections between
the eigenvectors and the indicator vectors of the clusters.
This characterization allows us to show that many variants of spectral clustering, that are based on the spectral embedding and that work ``in practice'', can be rigorously analyzed ``in theory''. Moreover, exploiting our gap assumption, we can approximate this spectral embedding using the heat kernel of the graph. Combining this with approximate nearest neighbor data structures, we give a
nearly-linear time algorithm for the $k$-way partitioning problem.

Our structural results can be summarized as follows.
 Let $\{f_i\}_{i=1}^k$   be the eigenvectors corresponding to the  $k$ smallest eigenvalues of $\calL$, and  $\{S_i\}_{i=1}^k$ be a $k$-way partition of $G$ achieving $\rho(k)$ defined in \eq{defrho}. We define $\{g_i\}_{i=1}^k$ to be the indicator vectors of the clusters 
$\{S_i\}_{i=1}^k$, where $g_i(u)=1$ if $u \in S_i$, and $g_i(u)=0$ otherwise. We further use $\{\bar{g}_i\}_{i=1}^k$ to express the normalized indicator vectors of the clusters $\{S_i\}_{i=1}^k$, defined by
\[
\bar{g}_i=\frac{\mathbf{D}^{1/2}g_i }{ \| D^{1/2}g_i \|}.
\]
We show that,  under the condition of $\Upsilon(k)=\Omega(k^2)$, the span of $\{\bar{g}_i\}_{i=1}^k$ and the span of $\{f_i\}_{i=1}^k$ are close to each other, which is stated formally in \thmref{STinf}.

\begin{thm}[The Structure Theorem]\label{thm:STinf}
 Let  $\{S_i\}_{i=1}^k$ be a $k$-way partition  of $G$ achieving  $\rho(k)$, and let $\Upsilon(k)=\lambda_{k+1}/\rho(k) = \Omega(k^2)$. Let $\{f_i\}_{i=1}^k$ and $\{\bar{g}_i \}_{i=1}^k$ be defined as above.
Then, the following statements hold:
\begin{enumerate}
\item For every $\bar{g}_i$, there is a
linear combination of $\{f_i\}_{i=1}^k$, called $\hat{f}_i$, such that
$
\|\overline{g}_i - \hat{f}_i\|^2 \le  1/\Upsilon(k).
$
\item For every $f_i$,  there is a linear combination of $\{\overline{g}_i\}_{i=1}^k$,
called $\hat{g}_i$,
  such that
$
\left\|f_i - \hat{g}_i\right\|^2 \le 1.1k/ \Upsilon(k).
$
\end{enumerate}
\end{thm}

This theorem generalizes  the result shown by Arora et al.~(\cite{ABS10}, Theorem 2.2),
which proves the easier direction~(the first statement, \cref{thm:STinf}), and can be  considered as a stronger  version of the well-known Davis-Kahan theorem~\citep{DavisKahan}. We remark that, despite that we use the higher-order Cheeger inequality  \eq{highorder}
to motivate the definition of $\Upsilon(k)$, our proof of the structure theorem is self-contained.
Specifically, it omits much of the machinery used in the proofs of higher-order and
improved Cheeger inequalities~\citep{impCheeger13,conf/stoc/LeeGT12}.

The structure theorem has several applications. For instance, we  look at the well-known  spectral embedding $F: V[G]\rightarrow\R^k$ defined by
\begin{equation}
\label{eq:speembed}
F(u) \triangleq \frac{1}{\mathsf{NormalizationFactor}(u)}\cdot\left(f_1(u),\dots,f_k(u)\right)^{\rot},
\end{equation}
where $\mathsf{NormalizationFactor}(u)\in\R$ is a normalization factor for  $u\in V[G]$.
We  use  \cref{thm:STinf} to show that
this well-known spectral embedding exhibits very nice geometric properties: (i)
\emph{all}  points $F(u)$ from the same cluster  are  close to each other, and (ii) \emph{most pairs of}  points $F(u), F(v)$ from  different clusters  are far from each other; (iii) the bigger the value of $\Upsilon(k)$, the higher concentration the embedded points within the same cluster.

Based on these facts, we analyze the performance of spectral clustering, aiming at answering the following longstanding open question:
\emph{Why does spectral clustering  perform well in practice?}
We show that the  partition $\{A_i\}_{i=1}^k$ produced  by spectral clustering gives  a good approximation of any ``optimal" partition $
\{S_i\}_{i=1}^k$: every $A_i$ has low  conductance, and has  large overlap with  its corresponding $S_i$. 
This algorithm has comprehensive applications, and has been the subject
of extensive experimental studies for more than 20 years, e.g. \citep{nips02,luxburg07}.
Prior to this work, similar results on spectral clustering mainly focus on graphs generated from the stochastic block model. Instead, our gap assumption captures more general classes of graphs
by replacing the input model with a structural condition.
Our result represents the first rigorous analysis of spectral clustering for the  general family of graphs that exhibit a multi-cut structure but are not captured by the stochastic block model.
Our result  is as follows:

\begin{thm}[Approximation Guarantee of Spectral Clustering] \label{thm:anlykmean}
Let $G$ be a graph satisfying the condition $\Upsilon(k)=\lambda_{k+1}/\rho(k)=\Omega(k^3)$, and $k\in \N$.
Let $F:V[G]\rightarrow \R^k$
be the embedding defined  in \eq{speembed}. Let $\aseq$ be a $k$-way  partition by any $k$-means algorithm running in $\R^k$ that achieves an approximation ratio $\APT$. Then,  the following statements hold: (i)
 $\vol(A_i\triangle S_i)= O\left(\APT\cdot k^3/\Upsilon(k)\right)\vol(S_i)$, and (ii)
$\phi_G(A_i)=1.1\cdot\phi_G(S_i)+O\left(\APT\cdot k^3/\Upsilon(k)\right)$.
\end{thm}


We further study fast algorithms for partitioning well-clustered graphs. Notice that, for
 moderately large values of $k$, e.g. $k=\omega(\log n)$, directly applying $k$-means algorithms and \cref{thm:anlykmean} does not give a nearly-linear time algorithm, since (i) obtaining the spectral embedding \eq{speembed}  requires $\Omega(mk)$ time for computing  $k$ eigenvectors, and (ii)  most $k$-means algorithms run in 
$\Omega(nk)$ time.

To overcome the first obstacle, we study the so-called \emph{heat kernel embedding} $x_t: V[G]\rightarrow \mathbb{R}^n$, an embedding from $V$ to $\mathbb{R}^n$ defined by
\[
x_t(u)\triangleq\frac{1}{\mathsf{NormalizationFactor}(u)}\cdot \left( \mathrm{e}^{-t\cdot \lambda_1}f_1(u), \cdots, \mathrm{e}^{-t\cdot \lambda_n}f_n(u) \right)
\]
for some  $t\in\mathbb{R}_{\geq 0}$. The heat kernel of a graph is a well-studied mathematical concept and is related to, for example,  the study of random walks~\cite{SL97}.
We exploit the heat kernel embedding to approximate the squared-distance $\|F(u)-F(v)\|^2$ of the embedded
points $F(u)$ and $F(v)$ via their  \emph{heat-kernel distance} $\|x_t(u)  - x_t(v) \|^2$. Since  the heat kernel distances between  vertices can be approximated in nearly-linear time~\cite{osv12}, this
 approach avoids the  computation of eigenvectors for a large value of $k$.
For the second obstacle,  instead of applying $k$-means algorithms as a black-box, we apply  approximate nearest-neighbor data structures. This can be viewed as an ad-hoc version of a $k$-means algorithm, and indicates that in many scenarios the standard
Lloyd-type heuristic widely used  in $k$-means algorithms
can eventually be avoided. Our  result is as follows:

\begin{thm}[Nearly-Linear Time Algorithm For Partitioning Graphs]
\label{thm:main_informal}
Let $G=(V,E)$ be a graph of $n$ vertices and $m$ edges, and $k=\omega(\log n)$ be the number of clusters. Assume that
$\Upsilon(k)=\lambda_{k+1}/\rho(k)= \tilde{\Omega}(k^5)$, and
 $\sseq$ is a $k$-way partition  such that $\phi_G(S_i)\leq \rho(k)$. Then there is an algorithm which runs  in $\tilde{O}(m)$ time and outputs a $k$-way partition $\aseq$ such that  (i)
 $\vol(A_i\triangle S_i)= \tilde{O}\left( k^4 / \Upsilon(k)\right) \vol(S_i)$, and (ii)
 $\phi_G(A_i)=1.1\cdot\phi_G(S_i) + \tilde{O}\left(k^4 / \Upsilon(k)\right)$.
The  $\tilde{O}(\cdot)$ and $\tilde{\Omega}(\cdot)$ terms here hide a factor of $\mathrm{poly}\log n$.
 \end{thm}

We remark that bounds of other expansion parameters of $k$-way partitioning can be derived from our analysis as well. For instance,
it is easy to see that $\rho(k)$ and  the \emph{normalized cut}~\citep{ShiM00} studied in machine learning, which is defined as the sum of the conductance of all returned clusters, differ by at most a factor of $k$, and the normalized cut value of a $k$-way partition from spectral clustering can be derived from our results.



\subsection{Related Work}

In the broadest sense, our algorithms are clustering routines.
Clustering can be formulated in many ways, and the study of
algorithms in many such formulations are areas of active work
~\cite{AwasthiBO12,Ben-David13,KannanVV04,MakarychevMV15}.
Among these, our work is most closely related to spectral clustering,
which is closely related to normalized or low conductance cuts~\cite{ShiM00}.
The $k$-way expansion that we study is always within a factor of $k$
of $k$-way normalized cuts.

Theoretical studies of  graph partitioning are often based
on augmenting the fractional relaxation of these cut problems
with additional constraints in the form of semidefinite programs
or Lasserre hierarchy.
The goal of our study is to obtain similar bounds using
more practical tools such as $k$-means and heat-kernel embedding.

%

Oveis Gharan and Trevisan~\cite{conf/soda/GharanT14} formulate the notion of  clusters with respect to the \emph{inner} and $\emph{outer}$ conductance: a cluster $S$ should have  low outer conductance, and the conductance of the induced subgraph by $S$ should be high.  Under a gap assumption between $\lambda_{k+1}$ and $\lambda_k$, they
  present a polynomial-time algorithm which finds a $k$-way partition $\aseq$ that satisfies the inner- and outer-conductance condition. In order to ensure that every $A_i$ has high inner conductance, they assume that $\lambda_{k+1}\geq\mathrm{poly}(k)\lambda_k^{1/4}$, which is much stronger than ours. Moreover, their algorithm runs in polynomial-time, in contrast to our nearly-linear time algorithm.


Dey et al.~\cite{Sidiropoulos14} studies the properties of the spectral embedding for graphs having a gap between $\lambda_k$ and $\lambda_{k+1}$ and presents a $k$-way partition algorithm, which is based on  $k$-center clustering and is similar in spirit to our work.  Using combinatorial arguments, they are able to show that the clusters concentrate around $k$ distant points in the spectral embedding. In contrast to our work, their result only holds for bounded-degree graphs, and cannot provide an approximate guarantee for individual clusters. Moreover, their algorithm runs in nearly-linear time only if $k=O(\mathrm{poly}\log n)$.

 We also explore the separation between $\lambda_k$ and $\lambda_{k + 1}$ from an algorithmic
 perspective, and show that this assumption interacts well with heat-kernel embeddings.
The heat kernel has been used in previous algorithms on local partitioning~\citep{Chung09}, balanced separators~\citep{osv12}.
It also plays a key role in current efficient approximation algorithms for
finding low conductance cuts~\citep{OrecchiaSVV08,Sherman09}.
However, most of these theoretical guarantees are through the matrix
multiplicative weights update framework~\citep{AroraHK12,AroraK07}.
Our algorithm instead directly uses the heat-kernel embedding to
find low conductance cuts.

There is  also a considerable amount of research on partitioning random graphs. For instance, in the Stochastic Block Model~(\textsf{SBM})~\citep{mcsherry2001spectral}, the input graph with $k$ clusters is generated according to probabilities $p$ and $q$ with $p>q$: an edge between   any  two vertices  within the same cluster is placed  with probability $p$, and an edge between any two vertices  from different clusters is placed with probability $q$. It is proven that spectral  algorithms  give the correct clustering for certain ranges of $p$ and $q$~\citep{mcsherry2001spectral,rohe2011spectral,vu2014}. However, the analysis of these algorithms cannot be easily generalized into our setting: we consider graphs where edges are not necessarily chosen independently with certain probabilities, but can be added in an ``adversarial" way.
For this reason, standard perturbation theorems used in the analysis of algorithms for \textsf{SBMs}, such as the Davis-Kahan theorem~\citep{DavisKahan}, cannot be always applied, and ad-hoc arguments specific for graphs, like our structure theorem~(\thmref{STinf}), become necessary.

\section{Preliminaries\label{sec:prelimininaries}}

Let $G=(V,E)$ be an undirected and unweighted graph with $n$ vertices and $m$ edges. 
The set of neighbors of a vertex $u$ is represented by $N(u)$,
and its degree is $d_u= |N(u)|$.
For any set $S\subseteq V$, let $\vol(S)\triangleq\sum_{u\in S} d_u$.
For any set $S, T\subseteq V$, we define
$E(S,T)$ to be the set of edges between $S$ and $T$,
aka $E(S, T) \triangleq \{\{u,v\}| u\in S\mbox{\ and\ } v\in T\}$.
For simplicity, we write $\partial S=E(S,V\setminus S)$ for any set $S\subseteq V$.
For two sets $X$ and $Y$, the symmetric difference of $X$ and $Y$ is defined
as $X\triangle Y\triangleq (X\setminus Y)\cup (Y\setminus X)$.

We work extensively with algebraic objects related to $G$.
We  use $\mat{D}$ to denote the $n\times n$ diagonal matrix
with $\mat{D}_{uu}=d_u$ for $u\in V[G]$. 
The \emph{Laplacian matrix} of $G$ is defined by $\mat{L}\triangleq \mat{D}-\mat{A}$, where $\mat{A}$ is the adjacency matrix of $G$
defined by $\mat{A}_{u,v}=1$ if $\{u,v\}\in E[G]$, and $\mat{A}_{u,v}=0$ otherwise.  The \emph{normalized Laplacian matrix} of $G$ is defined by
$
\calL\triangleq \mat{D}^{-1/2}\mat{L}\mat{D}^{-1/2} = \mat{I}-\mat{D}^{-1/2}\mat{A}\mat{D}^{-1/2}$.
For this matrix, we denote its $n$ eigenvalues with
$ 0=\lambda_1\leq\cdots\leq\lambda_n\leq 2$,
with their corresponding orthonormal eigenvectors  $f_1,\ldots, f_n$.
Note that if $G$ is connected, the first eigenvector is
$f_1=\mat{D}^{1/2}f$, where $f$ is any non-zero constant vector. 

For a vector $x\in \mathbb{R}^n$, the Euclidean norm
of  $x$ is given by
$
\|x\|=\left(\sum_{i=1}^n x_i^2\right)^{1/2}$.
For any $f:V\rightarrow \mathbb{R}$ and $h\triangleq \mathbf{D}^{-1/2}f$, the \emph{Rayleigh quotient} of $f$ with
respect to graph $G$ is  given by
\[
\mathcal{R}(f)\triangleq \frac{f^{\rot} \calL f}{\|f\|^2} = \frac{h^{\rot} \mat{L} h}{\|h\|_{\mat{D}}} = \frac{\sum_{\{u,v\}\in E(G)} \left(h(u)-h(v)\right)^2}{\sum_{u} d_u h(u)^2},
\]
where $
\|h\|_{\mat{D}} \triangleq h^{\rot} \mat{D} h$. 
Based on the Rayleigh quotient,  the conductance of a set $S_i$ can be expressed as  $\phi_G(S_i)=\mathcal{R}(\bar{g}_i)$,
and  the gap $\Upsilon(k)$ can be written as 
\begin{equation}\label{eq:def_Upsilon}
\Upsilon(k)=\frac{\lambda_{k+1}}{\rho(k)}=\min_{1\leq i\leq k}\frac{\lambda_{k+1}}{\phi_G(S_i)}=\min_{1\leq i\leq k}\frac{\lambda_{k+1}}{\mathcal{R}(\bar{g}_i)}.
\end{equation}
Since $k$ is always fixed as part of the algorithm's input, throughout the rest of the paper we always use $\Upsilon$ to express $\Upsilon(k)$ for simplicity. We will also  use $S_1,\ldots, S_k$  to express 
a $k$-way partition of $G$ achieving  $\rho(k)$.
Note that this partition may not be unique.


\section{Connection Between Eigenvectors and Indicator Vectors of Clusters\label{sec:structure}}


In this section we study the relations between the multiple cuts of a graph and the eigenvectors of the graph's normalized Laplacian matrix.
Given clusters $S_1 \ldots S_k$, define the indicator vector of cluster $S_i$ by
\begin{equation}
\label{eq:const_g}
g_i(u)=
\left\{
	\begin{array}{ll}
		1  & \mbox{if } u \in S_i, \\
		0 & \mbox{if } u \not\in {S_i},
	\end{array}
\right.
\end{equation}
and define the  corresponding normalized indicator vector  by
\begin{equation}
\overline{g}_i=\frac{\mat{D}^{1/2} g_i}{\|\mat{D}^{1/2}g_i\|}.
\end{equation}
A basic result in spectral graph theory states that $G$ has $k$ connected components if and only if the $k$ smallest eigenvalues are $0$,
 implying that the spaces spanned by  $f_1,\cdots, f_k$ and $\bar{g}_1,\cdots, \bar{g}_k$ are the same. Generalizing this result, we expect that these two spaces would be still similar if these $k$ components of $G$ are loosely connected, in the sense that (i) every eigenvector $f_i$ can be approximately expressed by a linear combination of  $\{\overline{g}_i\}_{i=1}^k$, and (ii) every indicator vector $\bar{g}_i$ can be approximately expressed by a linear combination of  $\{f_i\}_{i=1}^k$. This leads to our structure theorem, which is illustrated in \cref{fig:relation}.



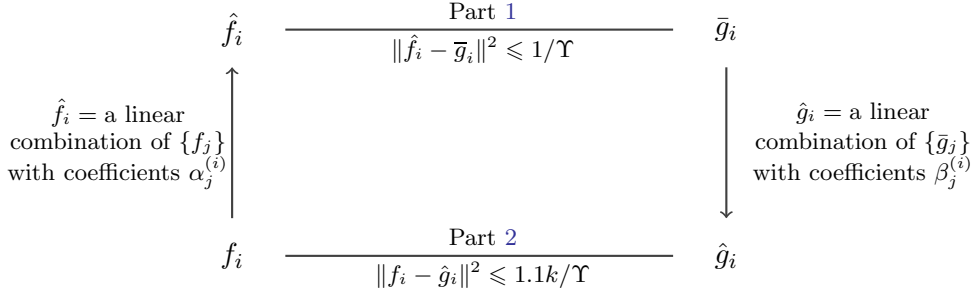
\begin{figure}[ht]\centering
\begin{tikzpicture}

\node (a) at (0, 0) {{$f_i$}};
\draw [->, color=darkgray,thick] (0,0.5) -- (0,2.5);
\node (b) at (0, 3) {{$\hat{f}_i$}};
\node (c) at (6.5, 0) {{$\hat{g}_i$}};
\node (d) at (6.5, 3) {{$\bar{g}_i$}};
\node (e) at (-1.5, 1.9) {\footnotesize{$\hat{f}_i=\mbox{a linear}$}};
\node (e) at (-1.5, 1.5) {\footnotesize{$\mbox{combination of } \{f_j\}$}};
\node (e) at (-1.5, 1.1) {\footnotesize{$\mbox{with coefficients}~\alpha^{(i)}_j$}};

\node (e) at (8.3, 1.9) {\footnotesize{$\hat{g}_i=\mbox{a linear}$}};
\node (e) at (8.3, 1.5) {\footnotesize{$\mbox{combination of } \{\bar{g}_j\}$}};
\node (e) at (8.3, 1.1) {\footnotesize{$\mbox{with coefficients}~\beta^{(i)}_j$}};
\draw [->, color=darkgray,thick] (6.5,2.5) -- (6.5,0.5);
\draw [ color=darkgray,thick] (0.7,0) -- (5.8,0);
\draw [ color=darkgray,thick] (0.7,3) -- (5.8,3);
\node (g) at (3.3, 3.25) {\footnotesize{Part~\ref{part:lincombK}}};
\node (h) at (3.3,2.75) {\footnotesize{$\|\hat{f}_i-\overline{g}_i  \|^2 \le1/\Upsilon$}};
\node (i) at (3.3, 0.25) {\footnotesize{Part~\ref{part:approxthm}}};
\node (k) at (3.3,-0.25) {\footnotesize{$\|f_i - \hat{g}_i\|^2 \le1.1 k/ \Upsilon$}};
\end{tikzpicture}
\caption{Relations among $\{\hat{f}_i\}$, $\{f_i\}$, $\{\bar{g}_i\}$, and $\{\hat{g}_i\}$ given in \cref{thm:structureFormal}.
Here $\Upsilon$ is the gap defined with respect to $\lambda_{k+1}$ and $\rho(k)$.\label{fig:relation} }
\end{figure}


\begin{thm}[The Structure Theorem, Formal Statement]\label{thm:structureFormal}
Let $\Upsilon=\Omega(k^2)$, and $1\leq i\leq k$. Then, the following statements hold:
\begin{enumerate}
\item \label{part:lincombK}
There is  a linear combination of the eigenvectors $f_1,\ldots, f_k$
with coefficients $\alpha_j^{(i)}$:
$
\hat{f}_i= \alpha^{(i)}_1f_1+\cdots + \alpha^{(i)}_kf_k,
$
such that
$
\left\|\overline{g}_i - \hat{f}_i\right\|^2 \le  1/\Upsilon$.

\item \label{part:approxthm}
There is a linear combination of the vectors $\bar{g}_1,\dots, \bar{g}_k$ with coefficients $\beta_j^{(i)}$:
$
\hat{g}_i=\beta_1^{(i)}\bar{g}_1+\cdots + \beta_k^{(i)}\bar{g}_k,
$
such that
$
\left\|f_i - \hat{g}_i\right\|^2 \le 1.1k/ \Upsilon$.
\end{enumerate}
\end{thm}

 Part~\ref{part:lincombK} of \cref{thm:structureFormal}  shows that  the normalized
indicator vectors $\bar{g}_i$ of every cluster $S_i$ can be approximated by a linear
combination of \emph{the first $k$ eigenvectors}, with respect to the value of $\Upsilon$.
The proof follows from the fact that if $\bar{g}_i$ has small Rayleigh quotient, then the inner product between $\bar{g}_i$
and the eigenvectors corresponding to larger eigenvalues must be small.  This statement was also shown implicitly in
 Theorem 2.2 of \cite{ABS10}.

\begin{proof}[Proof of Part~1 of \cref{thm:structureFormal}]
We write $\overline{g}_i$ as a linear combination of the eigenvectors of $\calL$, i.e.
\[
\overline{g}_i= \alpha^{(i)}_1f_1+\cdots + \alpha^{(i)}_nf_n
\]
and let the vector $\hat{f}_i$ be the projection of vector $\bar{g}_i$ on the subspace
spanned by  $\{f_i\}_{i=1}^k$, i.e.
\[
\hat{f}_i= \alpha^{(i)}_1f_1+\cdots + \alpha^{(i)}_kf_k.
\]
By  the definition of Rayleigh quotients, we have that
\begin{align*}
 \calR(\overline{g}_i) 	&=\left(\alpha^{(i)}_1 f_1 +\dots + \alpha^{(i)}_n f_n\right)^{\rot} \calL \left(\alpha^{(i)}_1 f_1 +\dots + \alpha^{(i)}_n  f_n\right) \\
								&=\left(\alpha^{(i)}_1\right)^2 \lambda_1 + \dots + \left(\alpha^{(i)}_n\right)^2 \lambda_n \\
								&\geq \left(\alpha^{(i)}_2\right)^2 \lambda_2 + \cdots + \left(\alpha^{(i)}_k\right)^2 \lambda_k + \left(1-\alpha' -\left(\alpha^{(i)}_1\right)^2\right)\lambda_{k+1},
\end{align*}
where $\alpha'\triangleq\left(\alpha^{(i)}_2\right)^2+\cdots +\left(\alpha^{(i)}_k\right)^2$. Therefore, we have that
\[
1-\alpha'-\left(\alpha^{(i)}_1\right)^2
\leq \mathcal{R}(\overline{g}_i)/\lambda_{k+1}\leq 1/\Upsilon,
\]
and
\[
\|\overline{g}_i - \hat{f}_i\|^2= \left(\alpha^{(i)}_{k+1}\right)^2+\cdots +\left(\alpha^{(i)}_n\right)^2= 1-\alpha'-\left(\alpha^{(i)}_1\right)^2  \leq 1/\Upsilon,
\]
which finishes the proof.
\end{proof}

 Part~\ref{part:approxthm} of \cref{thm:structureFormal} is more interesting, and shows that the opposite direction holds as well,
i.e.,  any $f_i~(1\leq i\leq k)$  can be approximated by a linear combination of
the normalized indicator vectors $\{\overline{g}_i\}_{i=1}^k$. To sketch the proof,  note that  if  we could write every
$\overline{g}_i$ \emph{exactly} as a linear combination of $\{f_i\}_{i=1}^k$,
then we could  write every $f_i~(1\leq i\leq k)$ as a linear combination
of  $\{\overline{g}_i\}_{i=1}^k$.
 This is because both of $\{f_i\}_{i=1}^k$ and $\{\overline{g}_i\}_{i=1}^k$
are sets of linearly independent vectors of the same dimension
and $\spn{\left\{\overline{g}_1,\dots,\overline{g}_k\right\}} \subseteq \spn{\{f_1,\dots,f_k\}}$.
However, the $\overline{g}_i$'s are only close to a linear
combination of the first $k$ eigenvectors, as shown in  Part~\ref{part:lincombK}.
We will denote this combination as $\hat{f}_i$, and use the fact that the errors
of approximation are small to show that these $\{\hat{f}_i\}_{i=1}^k$ are
almost orthogonal between each other.
This allows us to show that
$\spn{\{\hat{f}_1,\dots,\hat{f}_k\}} = \spn{\{f_1,\dots,f_k\}}$,
which  implies Part~\ref{part:approxthm}.

We will use the following two classical results in our proof.

\begin{thm}[Ger\v{s}gorin Circle Theorem]\label{thm:disc} Let $\mat{A}$ be an $n\times n$  matrix , and let $R_i(\mat{A})=\sum_{j\neq i} |\mat{A}_{i,j}|$, for $ 1\leq i \leq n.$
Then, all eigenvalues of  $\mat{A}$ are in the union of Ger\v{s}gorin Discs defined by
 \[
 \bigcup_{i=1}^n \left\{z\in \mathbb{C}: |z-\mat{A}_{i,i}|\leq R_i(\mat{A})\right\}.
 \]
\end{thm}

\begin{thm}[Corollary 6.3.4, \cite{HJ12}]
\label{thm:perturbation}
Let $\mat{A}$ be an $n\times n$ real and  symmetric matrix with eigenvalues $\lambda_1,\dots,\lambda_n$, and $\mat{E}$ be an $n\times n$ matrix. If $\hat{\lambda}$ is an eigenvalue of $\mat{A}+\mat{E}$, then there is some eigenvalue $\lambda_i$ of $\mat{A}$ for which $|\hat{\lambda}-\lambda_i|\le \|\mat{E}\|$.
\end{thm}

\begin{proof}[Proof of Part~\ref{part:approxthm} of \cref{thm:structureFormal}]
By Part~\ref{part:lincombK}, every $\overline{g}_i$ is approximated by a vector $\hat{f}_i$ defined by
\[
\hat{f}_i=\alpha^{(i)}_1 f_1+\cdots \alpha^{(i)}_k f_k.
\]

Define a $k$ by $k$ matrix $\mat{A}$ such that $\mat{A}_{i,j}=\alpha^{(j)}_i$, i.e.,
the $j$th column of matrix $\mat{A}$ consists of values $\left\{\alpha^{(j)}_i\right\}_{i=1}^k$ representing $\hat{f}_j$. We express the $j$th column of $\mat{A}$ by a vector $\alpha^{(j)}$, defined as
\[
\alpha^{(j)}= \left(\alpha^{(j)}_1,\cdots, \alpha^{(j)}_k\right)^{\rot}.
\]

We will show that the vectors $\left\{\alpha^{(j)}\right\}_{j=1}^k$ are linearly independent, which implies that $\{\hat{f}^{(j)}\}_{j=1}^k$ are linearly independent as well. 
To prove this, we will show that $\mat{A}^{\rot}\mat{A}$ has no zero eigenvalue, and hence 
  $\mat{A}$ is invertible. 

First of all, notice that it holds  by the orthonormality of $\{f_i\}_{i=1}^k$ that
\begin{align*}
\left|\left\langle \alpha^{(i)}, \alpha^{(j)} \right\rangle\right| &= \left|\left\langle \hat{f}_i, \hat{f}_j \right\rangle\right| 
			= \left|\left\langle \bar{g}_i - ( \bar{g}_i-\hat{f}_i),  \bar{g}_j - ( \bar{g}_j-\hat{f}_j) \right\rangle\right| \\
			&= \left| \left\langle  \bar{g}_i,  \bar{g}_j \right\rangle - \left\langle  \bar{g}_i-\hat{f}_i,  \bar{g}_j \right\rangle - \left\langle  \bar{g}_i, \bar{g}_j-\hat{f}_j  \right\rangle + \left\langle  \bar{g}_i-\hat{f}_i, \bar{g}_j-\hat{f}_j \right\rangle\right| \\
			&\le \left\|  \bar{g}_i-\hat{f}_i \right\| + \left\|  \bar{g}_j-\hat{f}_j \right\| + \left\|  \bar{g}_i-\hat{f}_i\right\| \left\|  \bar{g}_j-\hat{f}_j \right\| \\
								&\le 2\sqrt{1/ \Upsilon} + 1/ \Upsilon,
\end{align*}
where the first inequality follows from the orthonormality of $\bar{g}_i$ and $\bar{g}_j$, and the second inequality follows  by Part~\ref{part:lincombK} of \cref{thm:structureFormal}.
So it holds for any $i\ne j$ that
\[
\left|(\mat{A}^{\rot}\mat{A})_{i,j}\right| = \left|\sum_{\ell=1}^k \mat{A}_{\ell, i} \mat{A}_{\ell, j}\right| =
\left|\sum_{\ell=1}^k \alpha^{(i)}_{\ell} \alpha^{(j)}_{\ell}\right| = \left|\left\langle \alpha^{(i)}, \alpha^{(j)} \right\rangle \right| \leq
3\sqrt{1/ \Upsilon}
\]
while
$$(\mat{A}^{\rot}\mat{A})_{i,i} =\sum_{\ell=1}^k \left(\alpha^{(i)}_{\ell}\right)^2  \geq 1-1/\Upsilon.$$ Then,
 by the Ger\v{s}gorin Circle Theorem~(cf.~\cref{thm:disc}), it holds that  all the eigenvalues of $\mat{A}^{\rot}\mat{A}$ are at least
\[
1-1/\Upsilon - (k-1)\cdot 3\sqrt{1/\Upsilon}.
\]
Therefore, $\mat{A}$ has no eigenvalue with value $0$ as long as $\Upsilon>10 k^2$, proving that the vectors  $\left\{\alpha^{(j)}\right\}_{j=1}^k$ are linearly independent.
Combining this with the fact that $\spn{\{\hat{f}_1,\dots,\hat{f}_k\}} \subseteq \spn{\{f_1,\dots,f_k\}}$ and $\dmn( \spn{(\{f_1,\dots,f_k\})})= k$, it holds that 
\[
\spn{\{\hat{f}_1,\dots,\hat{f}_k\}} = \spn{\{f_1,\dots,f_k\}}.
\]
Hence,  we can write every $f_i~(1\leq i\leq k)$ as a linear combination of $\{\hat{f}_i\}_{i=1}^k$, i.e.,
\begin{equation}\label{eq:defbeta}
f_i=\beta^{(i)}_1\hat{f}_1+\beta^{(i)}_2\hat{f}_2+\cdots +\beta^{(i)}_k\hat{f}_k.
\end{equation}
Now define the value of $\hat{g}_i$ as
\begin{equation}\label{eq:defhatg}
\hat{g}_i=\beta^{(i)}_1\overline{g}_1+\beta^{(i)}_2\overline{g}_2+\cdots +\beta^{(i)}_k\overline{g}_k,
\end{equation}
 and define $\|\beta\|^2=\sum_{j=1}^k \left(\beta^{(i)}_j\right)^2$. Then, it holds that
\begin{align*}
1 = \|f_i\|^2 & = \sum_{\ell=1}^k \left( \beta^{(i)}_{\ell} \right)^2 \left\|\hat{f}_{\ell} \right\|^2 + \sum_{\ell\neq \ell'} \beta^{(i)}_{\ell} \beta^{(i)}_{\ell'} \left\langle \hat{f}_{\ell}, \hat{f}_{\ell'} \right\rangle\\
&
\ge \|\beta\|^2 (1-1/\Upsilon) -\sum_{\ell} \left| \beta^{(i)}_{\ell}\right| \sum_{\ell'\neq \ell} \left| \beta^{(i)}_{\ell'}\right| \left\langle \hat{f}_{\ell},  \hat{f}_{\ell'} \right\rangle\\
& \ge \|\beta\|^2 (1 -1/\Upsilon) -  \left(\sqrt{k}\cdot \|\beta\|\right) \cdot \left(\sqrt{k} \cdot \|\beta\|\right) \cdot \left(3\cdot \sqrt{1/\Upsilon}\right) \\
&\ge \left(1 -1/\Upsilon - 3k/\sqrt{\Upsilon} \right) \|\beta\|^2,
\end{align*}
where the second inequality holds by the Cauchy-Schwarz inequality. Since  $\Upsilon=\Omega(k^2)$, we have that 
\[
 \|\beta\|^2 \le \left(1-\frac{1}{\Upsilon}-\frac{3k}{\sqrt{\Upsilon}}\right)^{-1}\le 1.1.
\]
Combining this with Part~\ref{part:lincombK} of \cref{thm:structureFormal} and the Cauchy-Schwarz inequality, we have that 
\[
\| f_i-\hat{g}_i \| \leq \sum_{j=1}^k  \left|\beta_j^{(i)}\right| \left\|\hat{f}_j-\overline{g}_j \right\|\leq \left(1/\sqrt{\Upsilon}\right) \sum_{j=1}^k  \left|\beta_j^{(i)}\right| \le \sqrt{1.1 k/\Upsilon},
\]
which proves Part~\ref{part:approxthm} of the theorem.
\end{proof}

\cref{thm:structureFormal} shows a close connection between the
first $k$ eigenvectors and the indicator vectors of the clusters.
We leverage this and the fact that the $\{\hat{g}_i\}$'s are almost orthogonal between each other to show that, for any two different clusters $S_i$ and $S_j$, there exists
an eigenvector having reasonably different values on the coordinates which correspond to $S_i$ and $S_j$.

\begin{lem}\label{lem:deltalem}
Let $\Upsilon=\Omega(k^3)$.
For any $1\leq i\leq k$, let $\hat{g}_i=\beta_1^{(i)}\overline{g}_1+\cdots+\beta^{(i)}_k\overline{g}_k$ be such that
$
\|f_i - \hat{g}_i\| \le 1.1k/\Upsilon$.
Then, for any $\ell\neq j$, there  exists $i \in \{1,\dots,k\}$ such that
\begin{equation}\label{eq:defgamma}
\left|\beta^{(i)}_\ell - \beta^{(i)}_j\right| \ge\zeta\triangleq \frac{1}{10\sqrt{k}}.
\end{equation}
\end{lem}

\begin{proof}
Let $\beta^{(i)} =\left(\beta^{(i)}_1,\dots,\beta^{(i)}_k\right)^{\rot}$, for $1\leq i\leq k$.
Since $\bar{g}_{i} \perp \bar{g}_{j}$ for any $i \ne j$, we have by the orthonormality of $\overline{g}_1,\cdots,\overline{g}_k$  that
\begin{align*}
\left\langle \hat{g}_i , \hat{g}_j \right\rangle &= \left\langle \beta^{(i)}_1\overline{g}_1+\cdots +\beta^{(i)}_k\overline{g}_k , \beta^{(j)}_1\overline{g}_1+\cdots +\beta^{(j)}_k\overline{g}_k \right\rangle \\
						    &= \sum_{\ell=1}^k \beta^{(i)}_{\ell} \beta^{(j)}_{\ell} \|\overline{g}_{\ell}\|^2 = \left\langle \beta^{(i)} ,  \beta^{(j)} \right\rangle,
\end{align*}
 and
\begin{align*}
\left|\left\langle \beta^{(i)} ,  \beta^{(j)} \right\rangle \right|	&= \left| \left\langle \hat{g}_i , \hat{g}_j \right\rangle\right| = \left|\left\langle f_i - (f_i-\hat{g}_i), f_j - (f_j-\hat{g}_j) \right\rangle\right| \\
								&= \left|\langle f_i, f_j \rangle - \langle f_i-\hat{g}_i, f_j \rangle - \langle f_j-\hat{g}_j, f_i \rangle + \langle f_i-\hat{g}_i,f_j-\hat{g}_j \rangle\right| \\
								&\le \| f_i-\hat{g}_i \| + \| f_j-\hat{g}_j \| + \| f_i-\hat{g}_i \| \| f_j-\hat{g}_j \| \\
								&\le 2.2\sqrt{k/ \Upsilon} + 1.1k/ \Upsilon.
\end{align*}
Moreover, it holds that
\[
\left\| \beta^{(i)}\right\| = \left\| \hat{g}_i \right\| = \left\| f_i + \hat{g}_i  - f_i \right\| \leq 1 +  \left\|\hat{g}_i - f_i\right\| \leq 1+ \sqrt{1.1k/\Upsilon},
\]
and
\[
\left\| \beta^{(i)}\right\| = \left\| \hat{g}_i \right\| = \left\| f_i + \hat{g}_i  - f_i \right\| \geq 1 -  \left\|\hat{g}_i - f_i\right\| \geq 1- \sqrt{1.1k/\Upsilon},
\]
which implies that
\begin{equation}\label{eq:boundbeta}
\left\|\beta^{(i)}\right\|^2 \in \left(1 - (2.2\sqrt{k/ \Upsilon} + 1.1k/ \Upsilon), 1 +  2.2\sqrt{k/ \Upsilon} + 1.1k/ \Upsilon\right).
\end{equation}
In other words, we showed that $\beta^{(i)}$'s are almost orthonormal.

Now we construct a $k$ by $k$ matrix $\mat{B}$, where the $j$th column of $\mat{B}$ is $\beta^{(j)}$.
By the Ger\v{s}gorin Circle Theorem~(\cref{thm:disc}), all eigenvalues $\lambda$ of $\mat{B}^{\rot}\mat{B}$ satisfies
\begin{equation}\label{eq:aptb}
\left| \lambda-  (\mat{B}^{\rot}\mat{B})_{i,i} \right| \leq (k-1)\cdot ( 2.2\sqrt{k/ \Upsilon} + 1.1k/ \Upsilon)
\end{equation}
for any $i$. Combing this with \eq{boundbeta}, we have that the eigenvalues of  $\mat{B}^{\intercal}\mat{B}$ are close to 1.


Now we show that $\beta^{(i)}_\ell$ and $\beta^{(i)}_j$ are far
from each other by contradiction.
Suppose there exist $\ell \neq j$ such that
\[
\zeta' \triangleq \max\limits_{1\leq i\leq k} \left|\beta^{(i)}_\ell - \beta^{(i)}_j\right|  < \frac{1}{10\sqrt{k}}.
\]
 This implies that the $j$th row and $\ell$th row of matrix $\mat{B}$ are somewhat close to each other.
 Let us now define  matrix $\mat{E}\in\R^{k\times k}$, where
 \[
 \mat{E}_{\ell,i} \triangleq \beta^{(i)}_j - \beta^{(i)}_{\ell},
 \]
 and $\mat{E}_{t,i}=0$ for any $t\neq \ell$ and $1\leq i\leq k$. Moreover, let 
 \[
 \mat{Q}=\mat{B}+\mat{E}.
 \]
 Notice that $\mat{Q}$
  has two identical rows, and rank at most $k-1$. Therefore, $\mat{Q}$ has an eigenvalue with value $0$, and the spectral norm $\|\mat{E}\|$ of $\mat{E}$, the largest singular value of $\mat{E}$, is at most $\sqrt{k}\zeta'$. By definition of matrix $\mat{Q}$ we have that
\[\mat{Q}^{\rot}\mat{Q} =  \mat{B}^{\rot}\mat{B} +  \mat{B}^{\rot}\mat{E}+ \mat{E}^{\rot}\mat{B}+ \mat{E}^{\rot}\mat{E}.\]
Since $\mat{B}^{\rot}\mat{B}$ is symmetric and $0$ is an eigenvalue of $\mat{Q}^{\rot}\mat{Q}$, by \cref{thm:perturbation} we know that, if $\hat{\lambda}$ is an eigenvalue of $\mat{Q}^{\rot}\mat{Q}$, then there is an  eigenvalue $\lambda$ of $\mat{B}^{\rot}\mat{B}$ such that
\begin{align*}
|\hat{\lambda}-\lambda|& \leq \|\mat{B}^{\rot}\mat{E}+ \mat{E}^{\rot}\mat{B}+ \mat{E}^{\rot}\mat{E}\|\\
& \leq \|\mat{B}^{\rot}\mat{E}\| + \|\mat{E}^{\rot}\mat{B}\|+ \|\mat{E}^{\rot}\mat{E}\|\\
& \leq  4\sqrt{k}\zeta'+ k\zeta'^2,
\end{align*}
which implies  that
\[
\hat{\lambda}\geq \lambda-4\sqrt{k}\zeta'- k\zeta'^2\geq  1-k(2.2\sqrt{k/ \Upsilon} + 1.1k /\Upsilon)-4\sqrt{k}\zeta'- k\zeta'^2,
\]
due to \eq{boundbeta} and \eq{aptb}. By setting $\hat{\lambda}=0$,
we have that
\[
1-k(2.2\sqrt{k/\Upsilon}+1.1k/\Upsilon)-4\sqrt{k}\zeta' -k\zeta'^2 \leq 0.
\]
By the condition of  $\Upsilon=\Omega(k^3)$, the inequality above implies that $
\zeta' \ge \frac{1}{10\sqrt{k}}$, which leads to a contradiction.
\end{proof}

We point out that it was  shown in \cite{impCheeger13} that  the first $k$ eigenvectors can
be approximated by a $(2k + 1)$-step function.
The quality of the approximation is the same as the
one given by our structure theorem.
However,  a $(2k+1)$-step approximation is not enough to show that most vertices belonging to the same cluster are mapped close to each other in the spectral embedding.

We further  point out that  standard matrix perturbation theorems cannot be applied in our setting. For instance, we look at a well-clustered graph $G$ that  contains a subset $C$  of a cluster $S_i$  such that  most neighbors of vertices in $C$ are outside $S_i$. In this case, the adjacency matrix representing crossing edges of $G$ has high spectral norm, and hence  standard matrix perturbation arguments could not give us a meaningful result. However, our structure theorem takes the fact that $\vol(C)$ has to be small into account, and that is why the structure theorem is needed to analyze the cut structure of a graph.

\section{Analysis of Spectral Clustering\label{sec:withoutht}}

In this section we analyze an algorithm based on the classical spectral clustering paradigm, and give an approximation guarantee of this method on well-clustered graphs. We will show that  any $k$-means algorithm $\AlgoMean(\calX,k)$
with certain approximation guarantee can be used for  the $k$-way partitioning problem.
Furthermore, it suffices to call $\AlgoMean$ in a black-box manner with a point set $\calX \subseteq \mathbb{R}^{k}$.

 This section is structured as follows. We first give a quick overview of  spectral and $k$-means clustering in \secref{introkmean}. In \secref{analyembed}, we use the structure theorem to analyze the spectral embedding. \secref{analysis1} gives a general result about  $k$-means  when applied to this embedding, and the proof of \cref{thm:anlykmean}.

\subsection{$k$-Means Clustering\label{sec:introkmean} }

 Given a set of points $\calX\subseteq \R^d$, a \emph{k-means algorithm} $\AlgoMean(\calX,k)$ seeks to find a set $\mathcal{K}$ of $k$ centers $c_1,\cdots, c_k$ to minimize
  the sum  of the $\ell_2^2$-distance between $x\in\calX$ and the center to which it is assigned. Formally, for any partition $\calX_1,\cdots, \calX_k$ of the set $\calX\subseteq \R^d$, we define the cost function by
\[
\COST(\calX_1,\dots,\calX_k) \triangleq \min_{c_1,\dots,c_k \in \R^d} \sum_{i=1}^k\sum_{x \in \calX_i}  \| x-c_i\|^2,
\]
i.e.,
the $\COST$ function minimizes the total $\ell_2^2$-distance between the points $x$'s and their individually closest center  $c_i$, where $c_1,\ldots, c_k$ are chosen arbitrarily in  $\R^d$.
We further define  the optimal clustering cost by
\begin{equation}\label{eq:defdeltak}
\Delta_k^2(\calX)\triangleq\min_{\substack{\small \mathrm{partition }~ \calX_1,\ldots, \calX_k}}\COST(\calX_1,\dots,\calX_k).
\end{equation}

Spectral clustering can be described as follows: (i) Compute the bottom $k$ eigenvectors $f_1,\cdots, f_k$ of the normalized Laplacian matrix\footnote{Other graph matrices~(e.g. the adjacency matrix, and the Laplacian matrix) are also widely used in practice. Notice that, with proper normalization, the choice of these matrices does not substantially influence the performance of $k$-means algorithms.} of graph $G$. (ii) Map every vertex $u\in V[G]$ to a point $F(u)\in\R^k$ according to
\begin{equation}
\label{eq:embedding}
F(u) = \frac{1}{\mathsf{NormalizationFactor}(u)}\cdot\left(f_1(u),\dots,f_k(u)\right)^{\rot},
\end{equation}
with a proper normalization factor $\mathsf{NormalizationFactor}(u)\in\R$ for each $u\in V$. (iii)  Let $\mathcal{X}\triangleq\{F(u): u\in V\}$ be the set of the embedded points from vertices in $G$. Run $\AlgoMean(\mathcal{X},k)$, and group the vertices of $G$ into $k$ clusters according to the output of $\AlgoMean(\mathcal{X},k)$.
This approach that combines a $k$-means algorithm with a spectral embedding has been widely used
in practice for a long time, although there was a lack of rigorous analysis of its performance
prior to our result.

\subsection{Analysis of the Spectral Embedding\label{sec:analyembed}}

The first step of spectral clustering  is to map vertices of a graph into points in Euclidean space, through the spectral embedding  \eq{embedding}. This subsection analyzes the properties of this embedding. Let us define the normalization factor to be
\[
\mathsf{NormalizationFactor}(u)\triangleq \sqrt{d_u}.
\]
We will show that
the embedding \eq{embedding} with the normalization factor above  has very nice properties: embedded points from the same cluster $S_i$ are concentrated around their center $c_i\in \R^k$, and
 embedded points from different clusters of $G$ are far from each other.
These
properties imply that a simple $k$-means algorithm is able to produce a good clustering\footnote{Notice that this embedding is similar with the one used in \cite{conf/stoc/LeeGT12}, with the only difference that $F(u)$ is not normalized and so it is not necessarily a unit vector. This difference, though, is crucial for our analysis.}.

We first   define $k$ points $p^{(i)}\in \R^k~(1\leq i\leq k)$, where
\begin{equation}
\label{eq:centres}
p^{(i)}\triangleq\frac{1}{\sqrt{\vol{(S_i)}}}\left(\beta^{(1)}_i,\dots,\beta^{(k)}_i\right)^{\rot}
\end{equation}
and the parameters $\{ \beta_i^{(j)} \}_{j=1}^k$ are defined in \cref{thm:structureFormal}.
We will show in \cref{lem:optcost} that all embedded points $\mathcal{X}_i\triangleq\{F(u): u\in S_i\}~(1\leq i\leq k)$ are concentrated around $p^{(i)}$. Moreover, we  bound  the total $\ell_2^2$-distance between points in   $\mathcal{X}_i$ and $p^{(i)}$, which is proportional to $1/\Upsilon$: the bigger the value of $\Upsilon$, the higher concentration the points within the same cluster have. Notice that we \emph{do not} claim that $p^{(i)}$ is the actual center of $\calX_i$. However, these approximated points $p^{(i)}$'s suffice for our analysis.

\begin{lem}\label{lem:optcost}
It holds that $$\sum_{i=1}^k\sum_{u \in S_i} d_u \left\|  F(u) - p^{(i)} \right\|^2 \le 1.1k^2/\Upsilon.$$
\end{lem}

\begin{proof}
Since $\hat{g}_j(u) = \sqrt{\frac{d_u}{\vol(S_i)}} \beta^{(j)}_i $  and $p^{(i)}_j=\frac{1}{\sqrt{\vol(S_i)}} \beta_i^{(j)}$
hold for any $1\leq j \leq k$ and $u\in S_i$ by definition, we have that 
\begin{align*}
\sum_{i=1}^k \sum_{u \in S_i} d_u \left(F(u)_j - p^{(i)}_j\right)^2 &
= \sum_{i=1}^k \sum_{u\in S_i} d_u \left( \frac{1}{\sqrt{d_u}} f_j(u) - \frac{1}{\sqrt{\vol(S_i)}} \beta^{(j)}_i \right)^2\\
& =\sum_{i=1}^k \sum_{u\in S_i}  \left(  f_j(u) - \sqrt{\frac{d_u}{\vol(S_i)}} \beta^{(j)}_i \right)^2 \\
& =\sum_{i=1}^k \sum_{u\in S_i}  \left(  f_j(u) - \hat{g}_j(u)\right)^2 \\
&= \left\| f_j-\hat{g}_j \right\|^2\\
&\le 1.1k/\Upsilon,
\end{align*}
where the last inequality follows from  \cref{thm:structureFormal}.
Summing over all $j$ for $1\leq j\leq k$ implies that
\[
\sum_{i=1}^k \sum_{u \in S_i} d_u \left\|  F(u) - p^{(i)} \right\|^2 = \sum_{i=1}^k \sum_{j=1}^k \sum_{u \in S_i} d_u \left(F(u)_j - p^{(i)}_j\right)^2 \le 1.1k^2/\Upsilon.\qedhere
\]
\end{proof}

The next lemma shows that the $\ell^2_2$-norm of $p^{(i)}$ is inversely proportional to the volume of $S_i$. This implies that embedded points from a big cluster are close to the origin, while embedded points from a  small cluster are far from the origin.

\begin{lem}\label{lem:norm}
It holds for every $1\leq i\leq k$ that
\[
\frac{99}{100\vol(S_i)} \le \left\|p^{(i)}\right\|^2 \le \frac{101}{100\vol(S_i)}.
\]
\end{lem}
\begin{proof}
By \eq{centres}, we have that
\[
\left\|p^{(i)}\right\|^2 = \frac{1}{\vol(S_i)}\left\| \left(\beta^{(1)}_i,\dots,\beta^{(k)}_i\right)^{\rot} \right\|^2.
\]
Notice that $p^{(i)}$ is just the $i$th row of the matrix $\mat{B}$ defined in the proof of \cref{lem:deltalem}, normalized by $\sqrt{\vol(S_i)}$. Since $\mat{B}$ and $\mat{B}^{\rot}$ share the same
singular values (this follows from the SVD decomposition), by \eq{aptb} the eigenvalues of  $\mat{B}\mat{B}^{\rot}$ are close to $1$. But since $(\mat{B}\mat{B}^{\rot})_{i,i}$ is equal
to the $\ell_2^2$-norm of the $i$th row of $\mat{B}$, we have that
\begin{equation}\label{eq:rangebeta}
\left\| \left(\beta^{(1)}_i,\dots,\beta^{(k)}_i\right)^{\rot} \right\|^2 \in  \left(1 - (2.2\sqrt{k/ \Upsilon} + 1.1k/ \Upsilon), 1 +  2.2\sqrt{k/ \Upsilon} + 1.1k/ \Upsilon\right),
\end{equation}
which implies  the statement.
\end{proof}

 We will further show in \cref{lem:deltaclusters} that these points $p^{(i)} (1\leq i\leq k)$ exhibit another excellent property: the
distance between $p^{(i)}$ and $p^{(j)}$ is inversely proportional to the volume of the \emph{smaller} cluster between $S_i$ and $S_j$.  Therefore, points in  $S_i$ of smaller  $\vol(S_i)$ are far from points  in $S_j$ of bigger $\vol(S_j)$. Notice that, if this were not the case,  a misclassification of a small fraction of  points in $S_j$ could introduce  a large error to $S_i$.

\begin{lem}\label{lem:deltaclusters}
For every $i \ne j$, it holds that
\[
\left\|p^{(i)} - p^{(j)}\right\|^2 \ge \frac{\zeta^2}{10\min{\{\vol(S_{i}),\vol(S_j)\}}},
\]
where $\zeta$ is defined in \eq{defgamma}.
\end{lem}


\begin{proof}
Let $S_i$ and $S_j$ be two arbitrary clusters.
By \cref{lem:deltalem}, there exists $1\leq  \ell \leq k$ such that \[\left| \beta^{(\ell)}_{i} - \beta^{(\ell)}_{j}  \right| \ge \zeta.\] By the definition of $p^{(i)}$ and $p^{(j)}$ it follows that
\[
\left|\left| \frac{p^{(i)}}{\|p^{(i)}\|} - \frac{p^{(j)}}{\|p^{(j)}\|}\right|\right|^2
\geq \left( \frac{\beta_i^{(\ell)}}{\sqrt{\sum_{t=1}^k \left(\beta_i^{(t)}\right)^2}} -\frac{\beta_j^{(\ell)}}{\sqrt{\sum_{t=1}^k \left(\beta_j^{(t)}\right)^2}} \right)^2.
\]
By \eq{rangebeta}, we know that
\[
\sqrt{\sum_{\ell=1}^k \left(\beta_j^{(\ell)}\right)^2}=\left\| \left(\beta^{(1)}_j,\dots,\beta^{(k)}_j\right)^{\rot} \right\| \in \left(1-\frac{\zeta}{10},1+\frac{\zeta}{10}\right).
\]
Therefore, we have that
\[
\left|\left| \frac{p^{(i)}}{\|p^{(i)}\|} - \frac{p^{(j)}}{\|p^{(j)}\|}\right|\right|^2  \ge \frac{1}{2}\cdot \left( \beta^{(\ell)}_{i} -  \beta^{(\ell)}_{j} \right)^2 \ge\frac{1}{2}\cdot \zeta^2,
\]
and
\[
\left\langle \frac{p^{(i)}}{\|p^{(i)}\|} , \frac{p^{(j)}}{\|p^{(j)}\|} \right\rangle \le 1-\zeta^2/4.
\]
Without loss of generality, we assume that  $\left\|p^{(i)} \right\|^2 \ge \left\|p^{(j)}\right\|^2$. By \cref{lem:norm}, it holds that
\[
\left\|p^{(i)} \right\|^2\geq \frac{9}{10\cdot\vol(S_i)},
\]
and
\[
\left\|p^{(i)} \right\|^2 \ge \left\|p^{(j)}\right\|^2 \geq\frac{9}{10\cdot \vol(S_j)}.
\]
Hence, it holds that
\[
\left\|p^{(i)} \right\|^2 \ge \frac{9}{10\min{\{\vol(S_{i}),\vol(S_j)\}}}.
\]
We can now finish the proof by considering two cases
based on $\left\|p^{(i)} \right\|$.

\emph{Case~1:}
Suppose that $\left\|p^{(i)} \right\| \ge 4 \left\|p^{(j)}\right\|$. We have that
\begin{align*}
\left\|p^{(i)} - p^{(j)}\right\|	\ge \left\|p^{(i)}\right\| - \left\|p^{(j)}\right\| \ge \frac{3}{4} \left\|p^{(i)}\right\|,
\end{align*}
which implies that
\[
\left\|p^{(i)} - p^{(j)}\right\|^2 \ge  \frac{9}{16}\left\|p^{(i)} \right\|^2 \geq \frac{1}{2\min{\{\vol(S_{i}),\vol(S_j)\}}}.
\]

\emph{Case~2:} Suppose $\left\|p^{(j)}\right\| = \alpha \left\|p^{(i)} \right\|$ for $\alpha \in (\frac{1}{4},1]$. In this case, we have that
\begin{align*}
\left\|p^{(i)} - p^{(j)}\right\|^2 	&= \left\|p^{(i)}\right\|^2 + \left\|p^{(j)}\right\|^2 -2 \left\langle \frac{p^{(i)}}{\|p^{(i)}\|} , \frac{p^{(j)}}{\|p^{(j)}\|} \right\rangle \left\|p^{(i)}\right\|\left\|p^{(j)}\right\| \\
					&\ge {\left\|p^{(i)} \right\|^2} + {\left\|p^{(j)}\right\|^2} -  2 (1-\zeta^2/4)\cdot{\left\|p^{(i)} \right\|\left\|p^{(j)}\right\|}\\
					&= (1+\alpha^2) {\left\|p^{(i)} \right\|^2} -  2 (1-\zeta^2/4)\alpha\cdot{\left\|p^{(i)} \right\|^2}\\
					&=(1+\alpha^2-2\alpha+\alpha\zeta^2/2) \cdot {\left\|p^{(i)} \right\|^2}\\
					&\ge \frac{\alpha\zeta^2}{2}\cdot \left\|p^{(i)}\right\|^2 \geq \zeta^2\cdot\frac{1}{10\min{\{\vol(S_{i}),\vol(S_j)\}}},
\end{align*}
and the lemma follows.
\end{proof}

\subsection{Approximation Guarantees of Spectral Clustering}
\label{sec:analysis1}

Now we analyze why spectral clustering performs well for solving the $k$-way partitioning problem.  We assume that  $A_1,\ldots, A_k$ is any  $k$-way partition  returned by a $k$-means algorithm with an approximation ratio of $\APT$.

We map every vertex $u$ to $d_u$ identical points in $\R^k$. This ``trick'' allows us to bound the volume of the overlap between the clusters retrieved by a $k$-means algorithm and the optimal ones.  For this reason we define the cost function of  partition $A_1,\ldots, A_k$ of $V[G]$ by
\[
\COST(A_1,\dots,A_k) \triangleq \min_{c_1,\dots,c_k \in \R^k} \sum_{i=1}^k\sum_{u \in A_i} d_u \| F(u)-c_i\|^2,
\]
and the optimal clustering cost is defined
 by
 $$
 \Delta_k^2\triangleq\min_{\mathrm{partition~} A_1,\ldots, A_k} \COST(A_1,\ldots, A_k).$$
i.e., we define the optimal clustering cost in the same way as in \eq{defdeltak}, except that we look at the embedded points from vertices of $G$ in the definition.
From now on, we always refer $\COST$ and $\Delta_k^2$ as the $\COST$ and optimal $\COST$ values of  points $\{F(u)\}_{u\in V}$, and for technical reasons every point is counted $d_u$ times.
The next lemma gives an upper bound to the cost of the optimal $k$-means clustering which depends on the gap $\Upsilon$

\begin{lem}\label{lem:boundopt}
It holds that $\Delta_k^2\le 1.1k^2 /\Upsilon$.
\end{lem}

\begin{proof}
Since $\Delta_k^2$ is obtained by minimizing over all partitions $A_1,\ldots, A_k$ and $c_1,\ldots, c_k$, we have 
\begin{align}
\label{eq:degcost}
\Delta_k^2  \leq \sum_{i=1}^k\sum_{u \in S_i} d_u \left\|  F(u) - p^{(i)} \right\|^2.\end{align}
Hence the statement follows by applying \cref{lem:optcost}.
\end{proof}

Since $A_1,\cdots, A_k$ is the output of a $k$-means algorithm with approximation ratio $\APT$,  by \cref{lem:boundopt} we have that $\COST(A_1,\ldots, A_k)\leq \APT\cdot 1.1k^2/\Upsilon$.
We will show that this upper bound of  $\APT\cdot 1.1k^2/\Upsilon$ suffices to show that  this approximate clustering  $A_1,\ldots, A_k$  is close to the ``actual" clustering $S_1,\ldots, S_k$, in the sense that (i) every $A_i$ has low conductance, and (ii)
  under a proper permutation $\sigma:\{1,\ldots, k\}\rightarrow \{1,\ldots, k\}$, the symmetric difference between $A_i$ and $S_{\sigma(i)}$ is small.
  The fact is proven by contradiction: If we could always find a set $A_i$ with high symmetric difference with its correspondence $S_{\sigma(i)}$, regardless of how we map $\{A_i\}$ to their corresponding $\{S_{\sigma(i)}\}$, then the $\COST$ value will be high, which contradicts to the fact that
   $\COST(A_1,\ldots, A_k)\leq \APT\cdot 1.1k^2/\Upsilon$. The core of of the whole contradiction arguments is the following technical lemma, whose proof will be presented in the next subsection.

\begin{lem}
\label{lem:approx_clusters}
Let $A_1,\dots,A_k$ be a partition of $V$.  Suppose that,  for every permutation of the indices $\sigma:\{1,\ldots,k\}\rightarrow\{1,\ldots, k\}$, there exists $i$ such that $\vol\left(A_i \triangle S_{\sigma(i)}\right) \ge 2\eps \vol\left(S_{\sigma(i)}\right)$ for $\eps \ge 10^5\cdot k^3/\Upsilon$, then $\COST(A_1,\dots,A_k)\ge 10^{-4}\cdot\eps/k$.
\end{lem}

\begin{proof}[Proof of \cref{thm:anlykmean}] Let $A_1,\dots,A_k$ be
a $k$-way partition that achieves an approximation ratio of $\APT$, and let
\[
\eps = \frac{2\cdot 10^5  \cdot k^3\cdot\APT}{\Upsilon}.
\]
We first show that there exists a permutation $\sigma$ of the indices such that
\begin{equation} \label{eq:sigmacondition}
\vol\left(A_i \triangle S_{\sigma(i)}\right) \le \varepsilon \vol(S_{\sigma(i)}),\qquad \mbox{for any\ } 1\leq i\leq k.
\end{equation}
Assume for contradiction that  for all permutation $\sigma$ 
  there is  $1\le i\leq k$ such that \[\vol(A_i \triangle S_{\sigma(i)}) > \varepsilon \vol\left(S_{\sigma(i)}\right).\]
This implies by  \cref{lem:approx_clusters} that
\[
\COST(A_1,\dots,A_k) \geq  10 \cdot\APT\cdot k^2 /\Upsilon,\]
which   contradicts to the fact that $A_1,\dots,A_k$ is an $\APT$-approximation to a $k$-way partition, whose corresponding $k$-means cost is at most  $1.1\cdot \APT\cdot k^2/\Upsilon$.

Now we assume that $\sigma:\{1,\cdots, k\} \rightarrow\{1,\cdots, k\}$ is the permutation satisfying \eq{sigmacondition}, and bound the conductance of every cluster $A_i$. For any $1\leq i\leq k$, the number of leaving edges of $A_{i}$ is upper bounded by
\begin{align*}
\left|\partial \left(A_{i}\right)\right| & \leq \left|\partial \left(A_{i}\setminus S_{\sigma(i)}\right)\right| +\left|\partial \left(A_{i}\cap S_{\sigma(i)}\right)\right|  \\
& \leq \left|\partial \left(A_{i}\triangle S_{\sigma(i)}\right)\right| + \left|\partial\left (A_{i}\cap S_{\sigma(i)}\right)\right|.
\end{align*}
Notice that $\left|\partial \left(A_{i}\triangle S_{\sigma(i)}\right)\right|  \leq \eps\vol\left(S_{\sigma(i)}\right)$ by our assumption on $\sigma$, and every node in $\left|\partial\left (A_{i}\cap S_{\sigma(i)}\right)\right|$ either belongs to $\partial S_{\sigma(i)}\setminus S_{\sigma(i)}$ or $\partial \left(A_{i}\triangle S_{\sigma(i)}\right)$, hence 
\begin{align*}
\left|\partial \left(A_{i}\right)\right| & \leq \eps\vol\left(S_{\sigma(i)}\right) +  \phi_G\left(S_{\sigma(i)}\right)\vol\left(S_{\sigma(i)}\right)+\eps\vol\left(S_{\sigma(i)}\right)\\
& = \left(2\eps+\phi_G\left(S_{\sigma(i)}\right)\right)\vol(S_{\sigma(i)}).
\end{align*}
On the other hand, we have that 
\[
\vol\left(A_{i}\right) \geq \vol \left(A_{i}\cap S_{\sigma(i)}\right) \geq (1-2\eps)\vol(S_{\sigma(i)}).
\]
Hence,
\begin{align*}
\phi_G(A_i) 	& \leq \frac{(2\eps+\phi_G(S_{\sigma(i)}))\vol(S_{\sigma(i)})}{(1-2\eps)\vol(S_{\sigma(i)})}\\
&=\frac{2\eps+\phi_G(S_{\sigma(i)})}{1-2\eps}\\
&\le 1.1\cdot\phi_G(S_{\sigma(i)}) + O(\APT\cdot k^3/\Upsilon).\qedhere
\end{align*}
\end{proof}

\subsection{Proof of \cref{lem:approx_clusters}}

It remains to show \cref{lem:approx_clusters}.
Our proof is  based on the following high-level idea:
suppose by contradiction that there is a cluster $S_j$ which is very different from every cluster $A_{\ell}$.
Then there is a cluster $A_i$ with significant overlap with two different clusters $S_{j}$ and $S_{j'}$ (\cref{lem:approx_clusters2}).
However,  we already proved in  \cref{lem:deltaclusters}  that any two clusters are far from each other.
This implies that the $\COST$ value of $A_1,\ldots, A_k$ is high, which leads to  a contradiction.

\begin{lem}
\label{lem:approx_clusters2}
Suppose for every permutation $\pi : \{1,\dots,k\} \rightarrow \{1,\dots,k\}$ there exists an index $i$ such that
$
\vol\left(A_i\triangle S_{\pi(i)}\right) \geq 2\eps \vol\left(S_{\pi(i)}\right)$.
Then, at least one of the following two cases holds:
\begin{enumerate}
\item for any index $i$ there are indices $i_1 \ne i_2$ and $\eps_i \ge 0$ such that
\[
\vol(A_i \cap S_{i_1}) \ge \vol(A_i \cap S_{i_2}) \ge \eps_i \min{ \{ \vol(S_{i_1}),\vol(S_{i_2}) \} },
\]
and
$\sum_{i=1}^k \eps_i \ge \eps$;
\item there exist indices $i,\ell$ and $\eps_j \ge 0$ such that, for $j \ne \ell$,
\[
\vol(A_i \cap S_{\ell}) \ge \eps_i   \vol(S_{\ell}), \text{   } \vol(A_i \cap S_{j}) \ge \eps_i  \vol(S_{\ell})
\]
and
$\sum_{i=1}^k \eps_i \ge \eps$.
\end{enumerate}
\end{lem}

\begin{proof}
Let $\sigma:\{1,\ldots, k\}\rightarrow\{1,\ldots, k\}$ be the function defined by
\[
\sigma(i)=\argmax_{1\leq j\leq k}\frac{\vol(A_i\cap S_j)}{ \vol(S_j)}.
\]
We  first assume that  $\sigma$ is one-to-one, i.e.~$\sigma$ is a permutation. By the hypothesis of the lemma, there exists an index $i$ such that $\vol(A_i\triangle S_{\sigma(i)}) \geq 2\eps \vol(S_{\sigma(i)})$. Without loss of generality, we  assume that  $i=1$ and $\sigma(j)=j$ for $j=1,\dots,k$.
Notice that
\begin{equation}\label{eq:expressa1a}
\vol\left(A_1\triangle S_1\right) = \sum_{j\neq 1} \vol\left(A_j \cap S_1\right) +\sum_{j\neq 1} \vol\left(A_1\cap S_j\right).
\end{equation}
 Hence, one of the summations on the right hand side of  \eq{expressa1a} is at least $\eps\vol\left(S_1\right)$.  Now the proof is based on the  case distinction.

\emph{Case~1:} Assume that $\sum_{j\neq 1} \vol\left(A_j \cap S_1\right)\geq \eps\vol(S_{1})$. We
define $\tau_j$ for $2\leq j\leq k$ to be
\[
\tau_j=\frac{\vol\left(A_j\cap S_{1}\right)}{\vol\left(S_1\right)}.
\]
We have that
\[
\sum_{j\neq 1}\tau_j\geq \eps,
\]
and by the definition of $\sigma$ it holds  that
\[
\frac{\vol\left(A_j\cap S_j\right)}{\vol\left(S_j\right)} \geq
\frac{\vol\left(A_j\cap S_1\right)}{\vol\left(S_1\right)}=\tau_j
\]
for $2\leq j\leq k$. Setting $\eps_j=\tau_j$ for $2\leq j\leq k$ and $\eps_1=0$ finishes the proof of Case~1.

\emph{Case~2:} Assume that
\begin{equation}\label{eq:condition1a}
\sum_{j\neq 1} \vol\left(A_1\cap S_j\right)\geq \eps\vol(S_1).
\end{equation}
Let us define $\tau'_j$ for $1\leq j\leq k, j\neq 1$, to be
\[
\tau'_j=\frac{\vol(A_1\cap S_j)}{ \vol\left(S_1\right)}.
\]
Then,  \eq{condition1a} implies that
 \[
\sum_{j \ne 1} \tau'_{j} \ge \eps.
\]
The statement in this  case holds by assuming $\vol\left(A_1\cap S_{1}\right) \ge \eps\vol\left(S_{1}\right)$, since otherwise we have \[
\vol\left(  S_{1} \right) - \vol\left( A_1\cap S_{1}\right) =
\sum_{j\neq 1} \vol\left(A_j \cap S_{1}\right)\geq (1-\eps)\vol\left(S_{1}\right)
\geq \eps\vol\left(S_{1}\right),
\]
and this case was proven in Case~1.

So it suffices to  study the case in which $\sigma$ defined earlier is not one-to-one. Then, there is $j~(1\leq j\leq k)$ such that $j\not\in\{\sigma(1),\cdots, \sigma(k)\}$. For any $1\leq\ell\leq k$, let
\[
\tau_{\ell}''=\frac{\vol(A_{\ell}\cap S_j)}{\vol(S_j)}.
\]
Then, $\sum_{\ell=1}^k \tau''_{\ell}=1\geq \eps$ and it holds for any $1\leq \ell\leq k$ that
\[
\frac{\vol\left(A_{\ell} \cap S_{\sigma(\ell)}\right)}{\vol\left(S_{\sigma(\ell)}\right)} \geq \frac{\vol(A_{\ell}\cap S_j)}{\vol(S_j)}=\tau''_{\ell}.\qedhere
\]
\end{proof}

\definecolor{blue1}{rgb}{0.2, 0.2, 0.6}
\definecolor{red1}{rgb}{0.6, 0.0, 0.0}

\begin{proof}[Proof of \cref{lem:approx_clusters}]
We first  consider  the case when part 1 of \cref{lem:approx_clusters2} holds, i.e., for every $i$ there exist $i_1 \ne i_2$ such that
\begin{equation}
\label{eq:vol}
\begin{aligned}
& \vol(A_i \cap S_{i_1}) \ge \eps_i \min{\{\vol(S_{i_1}),\vol(S_{i_2})\}},  \\
& \vol(A_i \cap S_{i_2})\ge \eps_i \min{\{\vol(S_{i_1}),\vol(S_{i_2})\}},
 \end{aligned}
 \end{equation}
 for some $\eps\geq 0$,  and $\sum_{i=1}^k \eps_i \ge \eps$.

Let $c_i$ be the center of $A_i$. Let us assume without loss of generality that  $\|c_i-p^{(i_1)}\| \ge \|c_i-p^{(i_2)}\|$, which implies $\|p^{({i_1})} - c_i\| \ge \|p^{(i_1)} - p^{(i_2)}\|/2.$
However, points in  $B_i = A_i \cap S_{i_1}$ are far away from $c_i$, see \cref{fig:expproof}. We lower bound the value of $\COST(A_1,\dots,A_k)$ by only looking  at the contribution of points in the $B_i$s . Notice that by \cref{lem:optcost} the sum of the squared-distances between points in $B_i$ and $p^{(i_1)}$ is at most $k^2/\Upsilon$, while the distance between $p^{(i_1)}$ and $p^{(i_2)}$ is large (\cref{lem:deltaclusters}).
Therefore, we have that
\begin{align*}
\COST(A_1,\dots,A_k) = \sum_{i=1}^k \sum_{u\in A_i } d_u  \|F(u) - c_i\|^2 \ge \sum_{i=1}^k \sum_{u\in B_i } d_u  \|F(u) - c_i\|^2  \\
\end{align*}
By applying the inequality $a^2 + b^2 \ge (a-b)^2/2$, we have that


 \begin{align}
\COST(A_1,\dots,A_k) 	&\ge \sum_{i=1}^k \sum_{u\in B_i } d_u \left( \frac{\left\|p^{({i_1})} - c_i\right\|^2}{2} - \left\|F(u) - p^{(i_1)}\right\|^2\right)  \notag\\
	&\ge \sum_{i=1}^k \sum_{u\in B_i } d_u  \frac{\left\|p^{({i_1})} - c_i\right\|^2}{2} - \sum_{i=1}^k\sum_{u\in B_i } d_u \left\|F(u) - p^{(i_1)}\right\|^2 \notag \\
	&\ge \sum_{i=1}^k \sum_{u\in B_i } d_u  \frac{\left\|p^{({i_1})} - c_i\right\|^2}{2} - \frac{1.1k^2}{\Upsilon} \label{eq:app1} \\
	&\ge\sum_{i=1}^k \sum_{u\in B_i } d_u  \frac{\left\|p^{({i_1})} - p^{({i_2})}\right\|^2}{8} - \frac{1.1k^2}{\Upsilon} \notag \\
	&\ge \sum_{i=1}^k \frac{\zeta^2 \vol(B_i)}{ 80\min{\{\vol(S_{i_1}),\vol(S_{i_2})\}}} - \frac{1.1k^2}{\Upsilon} \label{eq:app2} \\
	&\ge \sum_{i=1}^k\frac{\zeta^2 \eps_i \min{\{\vol(S_{i_1}),\vol(S_{i_2})\}}}{ 80\min{\{\vol(S_{i_1}),\vol(S_{i_2})\}}} - \frac{1.1k^2}{\Upsilon} \notag \\
	&\ge \sum_{i=1}^k \frac{\zeta^2 \eps_i }{ 80} - \frac{1.1k^2}{\Upsilon} \notag \\
	&\ge \frac{\zeta^2 \eps }{ 80} - \frac{1.1k^2}{\Upsilon} \ge  \frac{\zeta^2 \eps }{100} \notag
\end{align}
where  \eq{app1} follows from  \cref{lem:optcost}, \eq{app2} follows from  \cref{lem:deltaclusters} and the last inequality follows from the assumption that  $\eps
\geq 10^5\cdot k^3/\Upsilon$.


 Now, suppose that  part 2 of \cref{lem:approx_clusters2} holds, i.e.~there are indices $i,\ell$ such that, for any $j \ne \ell$, it holds that
\begin{align*}
\vol(A_i \cap S_{\ell}) \ge \eps_i  \vol(S_{\ell}), \\
\vol(A_i \cap S_{j}) \ge \eps_i  \vol(S_{\ell})
\end{align*}
for some $\eps\geq 0$,  and
  $
 \sum_{i=1}^k \eps_i \ge \eps$.
 In this case, we only need to repeat the proof by setting, for any  $j \ne i$, $B_j = A_i \cap S_j$, $S_{j_1} = S_{\ell}$, and $S_{j_2} = S_j$.
\end{proof}

\begin{figure}\centering
\begin{tikzpicture}

\draw [color=gray,dotted, thick] (0.7,0.4) --(-3.3,-1.5);
\draw [color=gray,thick] (0.7,0.4) --(-2.6,-0.1);

\draw [color=gray,dotted, thick] (-3.3,-1.5) --(-2.6,-0.1);

\draw  [color=red1,thick] (-3.3,-1) ellipse (1.4cm and 2.6cm);
\draw [color=blue1,thick] (0,0) ellipse (3.5cm and 2cm);
\draw  [color=red1,thick] (3.3,-1) ellipse (1.4cm and 2.6cm);
\node (a) at (0, 2.4) {{$A_i$}};
\node (b) at (-4.5, 2) {{$S_{i_1}$}};
\node (c) at (4.5, 2) {{$S_{i_2}$}};
\node (d) at (-2.9, 0.05) {{$u$}};
\node (d) at (-2.7, 0.55) {{$B_i$}};

\fill (-3.3,-1.5) circle (2pt);
\fill (3.3,-1.5) circle (2pt);
\fill (0.7,0.4) circle (2pt);

\fill (-2.6,-0.1) circle (2pt);

\node (d) at (0.7, 0.7) {{$c_i$}};
\node (e) at (-3.3, -2) {{$p^{(i_1)}$}};
\node (f) at (3.3, -2) {{$p^{(i_2)}$}};
\end{tikzpicture}
\caption{We use the fact that $\left\|p^{(i_1)}-c_i\right\|\geq \left\|p^{(i_2)}-c_i\right\|$, and lower bound the value of $\COST$ function by only looking at the contribution of points  $u\in B_i$ for all $1\leq i\leq k$.\label{fig:expproof}}
\end{figure}
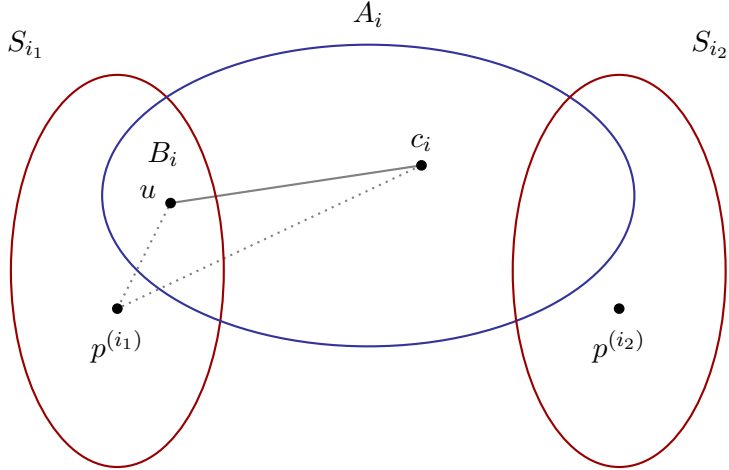

\section{Partitioning Well-Clustered Graphs in Nearly-Linear Time}
\label{sec:algorithm}

In this section we present a nearly-linear time algorithm for partitioning well-clustered graphs, and prove \cref{thm:main_informal}. At a high level, our algorithm follows the general framework of $k$-means algorithms, and consists of two steps: the seeding step, and the grouping step. The seeding step chooses $k$ candidate centers such that each one
is close to the actual center of a different cluster.
The grouping step assigns  the
remaining vertices to their individual closest candidate  centers.

All the proofs for the seeding and grouping steps assume that we have an embedding $\{x(u)\}_{u\in V[G]}$ satisfying the following two conditions:
\begin{align}
\left(1-\frac{1}{10\log{n}}\right) \cdot \| F(u) \|^2 \leq &\| x(u) \|^2 \leq  \| F(u) \|^2 + \frac{1}{n^5} \label{eq:approx_emb1}, \\
\left(1-\frac{1}{10\log{n}}\right) \cdot \| F(u)-F(v) \|^2 \leq &\| x(u) - x(v)\|^2 \leq  \| F(u)-F(v) \|^2 + \frac{1}{n^5}
\label{eq:approx_emb}
\end{align}

Notice that these two conditions hold trivially if $\{x(u)\}_{u\in V[G]}$ is the spectral embedding $\{ F(u)\}_{ u \in V[G]}$, or any embedding produced by good approximations of the first $k$ eigenvectors. However, obtaining such embedding becomes non-trivial 
for a large value of $k$,
as  directly computing the first $k$ eigenvectors  takes super-linear time. 
We will present a nearly-linear time algorithm that computes an embedding satisfying \eq{approx_emb1} and \eq{approx_emb}. By using standard dimensionality reduction techniques that approximately preserve pairwise distances, such as the Johnson-Lindenstrauss  transform (see e.g.~\cite{dasguptaJL}),
 we can also always assume that the dimension of the embedding $\{x(u)\}_{u\in V[G]}$ is $d=O(\log^3 n)$.
Throughout the whole section, we assume $k=\omega(\poly\log n)$ and
$\Upsilon=\tilde{\Omega}(k^5)$.

This section is organized as follows: \secref{seed} and \secref{group} discuss the seeding and grouping steps,
assuming that we have an embedding $\{x(u)\}_{u\in V[G]}$ that satisfies \eq{approx_emb1} and \eq{approx_emb}, and  \secref{approx}  analyzes  the approximation guarantee of the partition returned by the grouping step. In \secref{heatker}, we present an algorithm that computes all required quantities in nearly-linear time, assuming that we know the value of $\lambda_k$. This assumption on $\lambda_k$ will be finally removed in  Section~\ref{sec:pfthm13}, and this leads to our final algorithm which  corresponds to \cref{thm:main_informal}.

\subsection{The Seeding Step}
\label{sec:seed}
We proved in Section~\ref{sec:analyembed} that the approximate center $p^{(i)}$ for every $1\leq i\leq k$ satisfies
\[
\frac{99}{100\vol(S_i)} \le \left\|p^{(i)}\right\|^2 \le \frac{101}{100\vol(S_i)},
\]
and most embedded points $F(u)$ are close to their approximate centers. 
Together with \eq{approx_emb1} and \eq{approx_emb}, these two properties 
 imply that, when sampling points $x(u)$ with probability proportional to $d_u\cdot \|x(u)\|^2$,  vertices from different clusters will be approximately sampled with the same probability. We will prove that, when sampling $\Theta(k\log k)$ points in this way, with constant probability there is at least one point sampled from each cluster.

In the next step   remove the sampled points  which are close to each other, and call this resulting set $C^{\star}$. We prove that  with constant probability there is exactly one point in $C^{\star}$ from a cluster. Algorithm~\ref{fig:seed} below gives a formal description of the seeding step.

\begin{algorithm}

  \begin{algorithmic}[1]

\STATE \textbf{input}: the number of clusters $k$, and the embedding $\{x(u)\}_{u\in V[G]}$.

	\STATE Let $K=\Theta(k\log k)$.	
	 \FOR  {$i = 1,\dots, K$}
			\STATE Set $c_i = u$ with probability proportional to $d_u\| x(u) \|^2$.
	\ENDFOR
	 \FOR  {$i = 2,\dots,K$}
			\STATE Delete all $c_j$ with $j < i$ such that $\|x( c_i) - x(c_j )\|^2 < \frac{\| x(c_i) \|^2}{2\cdot10^4 k}$.
	\ENDFOR
	\RETURN the remaining sampled vertices.
 \end{algorithmic}
 \caption{\textsc{SeedAndTrim}$(k,\{x(u)\}_{u\in V[G]})$}
\label{fig:seed}
\end{algorithm}

Now we analyze Algorithm~\ref{fig:seed}.
For any $1\leq i\leq k$, we define $\mathcal{E}_i$ to be the sum of the $\ell_2^2$-distance between the embedded points from $S_i$ and $p^{(i)}$, i.e.,
\[
\mathcal{E}_i \triangleq\sum_{u\in S_i} d_u\left\|F(u)-p^{(i)} \right\|^2.
\]
For any parameter $\rho>0$, we define the radius of $S_i$ with respect to $\rho$ to be 
\[
R_i^{\rho}\triangleq \frac{\rho \cdot \mathcal{E}_i}{\vol(S_i)},
\]
and define $\core_i^{\rho}\subseteq S_i$ to be the set of vertices whose $\ell_2^2$-distance to $p^{(i)}$ is at most $R_i^{\rho}$, i.e.,
\begin{equation}
\label{eq:defcore}
\core_i^{\rho} \triangleq\left\{u\in S_i : \left\|F(u) - p^{(i)} \right\|^2 \leq R^{\rho}_i \right \}.
\end{equation}
By the averaging argument it holds that
\[
\vol(S_i\setminus \core_i^{\rho}) \leq \frac{\sum_{u\in S_i} d_u \left\| F(u)-p^{(i)}\right\|^2 }{ R_i^{\rho}}  =\frac{\vol(S_i)}{\rho},
\]
and 
\begin{equation}\label{eq:lbcore}
\vol(\core_i^{\rho})\geq \max\left\{\left(1-\frac{1}{\rho}\right)\vol(S_i), 0\right\}.
\end{equation}
 From now on, we set the parameter
 \[
\alpha\triangleq\Theta(K\log K).
\]
We first show that most embedded points of the vertices in $S_i$ are contained in the cores $\core_i^{\alpha}$, for $1\leq i\leq k$.
\begin{lem}
\label{lem:exactProb}
The following statements hold:
\begin{enumerate}
\item
$
\sum_{u\in\core_{i}^{\alpha}} {d_u\cdot \|F(u)\|^2} \geq 1-\frac{1}{100K} .
$
\item
$
\sum_{i=1}^k \sum_{u \notin \core_{i}^{\alpha} } {d_u\cdot \|F(u)\|^2} \leq    \frac{k}{100 K} .
$
\end{enumerate}
\end{lem}
\begin{proof}
By the definition of $\core_i^{\alpha}$, it holds that
\begin{align}
& \sum_{u \in \core_{i}^{\alpha} } {d_u\cdot \|F(u)\|^2} \nonumber \\
& \geq \frac{1}{\alpha} \int_0^{\alpha} \sum_{u \in \core_{i}^{\rho} } {d_u\cdot \|F(u)\|^2} d\rho \notag \\
&\ge  \frac{1}{\alpha} \int_0^{\alpha}  \left(\left\|p^{(i)}\right\| - \sqrt{R_i^{\rho}}\right)^2 \vol(\core_i^{\rho}) d\rho \label{eq:int11} \\
&\ge  \frac{1}{\alpha}  \int_0^{\alpha}  \left(\left\|p^{(i)}\right\|^2 - 2\sqrt{R_i^{\rho}}\cdot\left\|p^{(i)}\right\|\right) \max\left\{\left(1-\frac{1}{\rho}\right)\vol(S_i),0\right\} d\rho \label{eq:int21}  \\
&\ge  \frac{1}{\alpha} \int_0^{\alpha}  \max\left\{\left(1 - (2.2\sqrt{k/ \Upsilon} + 1.1k/ \Upsilon) - 3  \sqrt{ \mathcal{E}_i \rho}\right) \left(1-\frac{1}{\rho}\right),0\right\}d\rho \label{eq:int31}
\end{align}
where \eq{int11} follows from the fact that for all $u \in \core_i^{\rho}$, $\|F(u)\| \ge \|p^{(i)}\| - \sqrt{R_i^{\rho}}$, \eq{int21} from 
\eq{lbcore}, and \eq{int31} from the definition of $R_i^{\rho}$ and the fact that
\[
\left \|p^{(i)}\right\|^2\cdot\vol(S_i)\in  \left(1 - (2.2\sqrt{k/ \Upsilon} + 1.1k/ \Upsilon), 1 +  2.2\sqrt{k/ \Upsilon} + 1.1k/ \Upsilon\right).
\]
Since $\mathcal{E}_i\leq 1.1k^2/\Upsilon$   by \cref{lem:optcost}, it holds that 
\begin{align}
\sum_{u \in \core_{i}^{\alpha} } {d_u\cdot \|F(u)\|^2}
&\ge  \frac{1}{\alpha} \int_0^{\alpha}  \max\left\{\left(1 - (2.2\sqrt{k/ \Upsilon} + 1.1k/ \Upsilon) - 4\sqrt{k^2\rho/\Upsilon}\right) \left(1-\frac{1}{\rho}\right),0\right\}d\rho \notag  \\
&\ge  \frac{1}{\alpha} \int_0^{\alpha}  \max\left\{1 - (2.2\sqrt{k/ \Upsilon} + 1.1k/ \Upsilon) - 4\sqrt{k^2\rho/\Upsilon} -\frac{1}{\rho} ,0\right\}d\rho \notag  \\
&\ge 1 - (2.2\sqrt{k/ \Upsilon} + 1.1k/ \Upsilon) - 4k\sqrt{\alpha/\Upsilon} -\frac{\ln{\alpha}}{\alpha} \notag \\
&\ge 1 - \frac{1}{100K},\nonumber
\end{align}
where the last inequality holds by the  assumption on $\alpha$ and $\Upsilon$.

The second statement follows by the fact that
\[
\sum_{i=1}^k \sum_{u \in \core_{i}^{\alpha} } {d_u\cdot \|F(u)\|^2} \geq k\left(1-  \frac{1}{100K}\right)
\]
and
$
\sum_{u\in V[G]} d_u\cdot\|F(u)\|^2=k$.
\end{proof}

The next lemma shows that the embedded  points from the same core are close to each other, while the embedded points from different cores are far from each other.

\begin{lem}
\label{lem:sameCoreSmall}
The following statements hold:
\begin{enumerate}
\item
For any $1\leq i\leq k$ and any two vertices $u, v \in \core_{i}^{\alpha}$, it holds that
\[
\left\|  x(u) - x(v) \right\|^2 \leq \min\left\{  \frac{11 \alpha k^2 }{ \Upsilon \vol(S_i)},  \frac{\|x(u)\|^2}{2 \cdot 10^4\cdot k} \right\}.
\]
\item For any $i \ne j$, and  $u \in \core_i^{\alpha}, v  \in \core_j^{\alpha}$, it holds that
\[
\left\|  x(u) - x(v) \right\|^2
\geq \frac{1}{7000 k \vol(S_i)}  > \frac{\|x(u)\|^2}{10^4 k}.
\]

\end{enumerate}
\end{lem}

\begin{proof}
By the definition of $\core_{i}^{\alpha}$, it holds for any $u\in \core_{i}^{\alpha}$ that
\[
\left\|F(u) - p^{(i)}\right \|\leq \sqrt{R_i^{\alpha}}.
\]
By the triangle inequality, it holds for any $u\in\core_i^{\alpha}$ and $v\in\core_i^{\alpha}$ that
$
\|F(u) - F(v) \| \leq 2 \sqrt{R_i^{\alpha}}$,
and
\[
\|F(u) - F(v) \|^2 \leq 4 R^{\alpha}_i = \frac{4 \alpha \mathcal{E}_i} { \vol(S_i)}
\leq \frac{5 \alpha k^2 }{ \Upsilon \vol(S_i)},
\]
where the last inequality follows from \cref{lem:optcost}. Hence, by \eq{approx_emb} it holds that
\begin{align*}
\|x(u) - x(v) \|^2
& \leq \|F(u) - F(v) \|^2 + \frac{1}{n^5} \leq \frac{5\alpha k^2 }{ \Upsilon \vol(S_i)} + \frac{1}{n^5} 
\leq
\frac{11\alpha k^2}{ \Upsilon \vol(S_i)},
\end{align*}
where we use the fact that $\frac{1}{n^5} \ll \frac{1}{\vol(S_i)}$.
On the other hand,   we have that 
\begin{align*}
\| F(u) \|^2
&\geq\left (\left\|p^{(i)} \right\| -\sqrt{R^{\alpha}_i} \right)^2  \ge \frac{9}{10 \vol(S_i)},
\end{align*}
where the last inequality follow from \cref{lem:norm} and the definition of $R_i^{\alpha}$. 
By 
\eq{approx_emb1} and 
the conditions on $\alpha$, $\Upsilon$,  it also holds
\[
\|x(u) - x(v) \|^2 \leq \frac{5\alpha k^2 }{ \Upsilon \vol(S_i)} + \frac{1}{n^5}\le \frac{10\alpha k^2}{ \Upsilon} \| F(u) \|^2 \leq \frac{\|x(u)\|^2}{2 \cdot 10^4\cdot k}.
\]
With these we proved the first statement.

Now for the second statement.
By the triangle inequality, it holds for any pair of $u \in \core_i^{\alpha}$ and $v \in \core_{j}^{\alpha}$ that
\[
\|  F(u) - F(v) \| \geq \left\|p^{(i)} - p^{(j)}\right\| - \left\| F(u) - p^{(i)} \right\|   - \left\| F(v) -  p^{(j)} \right\|.
\]
By \cref{lem:deltaclusters}, we have for any $i \ne j$ that
\[
\left\|p^{(i)} - p^{(j)}\right\|^2 \ge  \frac{1}{10^3 k\min{\{\vol(S_{i}),\vol(S_j)\}}}.
\]
Combining this with the fact  that
\[
\left\| F(u) - p^{(i)} \right\| \leq \sqrt{R_i^{\alpha}}\leq  \sqrt{\frac{1.1\alpha  k^2 }{ \Upsilon \vol(S_i)}},
\]
we obtain that
\begin{align*}
\|  F(u) - F(v) \| &\geq \left\|p^{(i)} - p^{(j)}\right\| - \left\| F(u) - p^{(i)} \right\|   - \left\| F(v) -  p^{(j)} \right\|\\
& \geq \sqrt{\frac{1}{10^3 k\min{\{\vol(S_{i}),\vol(S_j)\}}}} - \sqrt{\frac{1.1\alpha k^2 }{ \Upsilon \vol(S_i)}} -\sqrt{\frac{1.1\alpha k^2 }{ \Upsilon \vol(S_j)}}\\
&\geq  \sqrt{\frac{1}{1.1\cdot 10^3 k\min{\{\vol(S_{i}),\vol(S_j)\}}}}.
\end{align*}
Notice that
\[
\| x(u)\|^2\leq \| F(u) \|^2 + \frac{1}{n^5} \leq  \left(\left\|p^{(i)} \right\| +\sqrt{R^{\alpha}_i} \right)^2 + \frac{1}{n^5} \leq \frac{11}{10 \vol(S_i)}+\frac{1}{n^5} \leq \frac{11}{9 \vol(S_i)},
\]
therefore we have 
\begin{align*}
\left\|x(u) - x(v) \right\|^2
& \geq \left(1-\frac{1}{10\log{n}}\right) \left\|F(u) - F(v) \right\|^2 \geq \frac{1  }{7000 k \vol(S_{i})} > \frac{\|x(u)\|^2}{10^4 k}. \qedhere
\end{align*}
\end{proof}

We next show that, after sampling $\Theta(k\log k)$ vertices, with constant probability  the sampled vertices are in the cores  $\bigcup_{i=1}^k \core_i^{\alpha}$, and every core contains at least one sampled vertex.

\begin{lem}
\label{lem:success}
Assume that $K=\Omega(k\log k)$ vertices are sampled, in which each vertex is sampled with probability proportional to $d_u\cdot \|x(u)\|^2$.
Then,  with constant probability the set $C = \{ c_1 \ldots c_K \}$ of sampled vertices satisfies the following properties:
\begin{enumerate}
\item Set $C$ only contains vertices from the cores, i.e.
$
C\subseteq \bigcup_{i=1}^k \core_i^{\alpha}
$;
 \item Set $C$ contains at least one vertex from each cluster, i.e.
 $
C\cap S_i \neq \emptyset$ for any $1 \leq i \leq k$.

 \end{enumerate}
\end{lem}

\begin{proof}
By \eq{approx_emb1}, it holds for every vertex $u$ that
\[
\left(1-\frac{1}{10\log{n}}\right) \cdot \| F(u) \|^2 \leq \| x(u) \|^2 \leq  \| F(u) \|^2 + \frac{1}{n^5}.
\]
Since
$
\sum_{u\in V[G]} d_u\|F(u)\|^2=k,
$
it holds that 
\[
\sum_{u\in V[G]} d_u\|x(u)\|^2\leq\sum_{u\in V[G]} d_u\cdot\left( \| F(u) \|^2 + \frac{1}{n^5}\right)\leq k+1,
\]
and
\[
\sum_{u\in V[G]} d_u\|x(u)\|^2\geq \sum_{u\in V[G]} d_u\cdot \left(1-\frac{1}{10\log{n}}\right) \cdot \| F(u) \|^2 \geq  \left(1-\frac{1}{10\log{n}}\right) \cdot k,
\]
i.e., 
the total probability mass that we use to sample vertices, i.e. $\sum_{u\in V[G]} d_u\|x(u)\|^2$, is between $\left(1-\frac{1}{10\log{n}}\right)\cdot k$ and $k+1$.

We first bound the probability that we sample at least one vertex from
every core.
For any fixed $1\le i \le k$, the probability that a vertex from $\core_i^{\alpha}$ gets sampled is at least
\begin{align*}
\frac{\sum_{u\in\core_{i}^{\alpha}} d_u\cdot \|x(u)\|^2}{k + 1}
 \geq \frac{\sum_{u\in\core_{i}^{\alpha}} d_u\cdot \|F(u)\|^2}{3(k + 1)}
 \geq \frac{1-\frac{1}{100 K}  }{3\cdot (k + 1)} \geq \frac{1}{10k}.
\end{align*}
Therefore, the probability that we never encounter a vertex from $\core_i^{\alpha}$ after
sampling $K$ vertices  is at most
$
\left( 1 - \frac{1}{10k} \right)^{K}
\leq \frac{1}{10k}$.
Also, the probability that
a sampled vertex is  outside the cores of the clusters is at most
\begin{align*}
 \frac{\sum_{i=1}^k\sum_{u\in S_i\setminus \core_{i}^{\alpha}} d_u\cdot \|x(u)\|^2}{\left(1-\frac{1}{10\log{n}}\right)\cdot k}
& \leq \frac{\sum_{i=1}^k\sum_{u\in S_i\setminus \core_{i}^{\alpha}} d_u\cdot\left( \|F(u)\|^2 + n^{-5}\right) }{k /2}\\
& \leq \frac{\frac{k}{100 K}  + n^{-3}}{k/2} \leq   \frac{2}{100K} + \frac{1}{n^2}.
\end{align*}
Taking a union bound over all these events gives that the total probability
of undesired events is at most 
\[
k \cdot \frac{1}{10k} + K \cdot \left(\frac{1}{n^2} + \frac{2}{100 K}\right)
\leq \frac{1}{3}.\qedhere
\]

\end{proof}

Based on
\cref{lem:sameCoreSmall} and \cref{lem:success}, we can simply delete one of the two vertices  $c_i$ and $c_j$ whose distance is less than $10^{-4}\cdot\|x(c_i)\|^2/(2k)$.  The following lemma presents the correctness and runtime of the procedure
\textsc{SeedAndTrim}, i.e., Algorithm~\ref{fig:seed}.

\begin{lem}
\label{lem:seedAndTrim}
Given the embedding $\{x(u)\}_{u\in V[G]}$ of dimension $d=O(\log^3n)$ that satisfies \eq{approx_emb1} and  \eq{approx_emb}, with constant probability the
procedure
\textsc{SeedAndTrim} returns  a set $C^{\star}$ of
centers $c_1 \ldots c_k$ in $\tilde{O}(n + k^2)$ time, such that each $\core_{i}^{\alpha}$ contains
exactly one vertex in $C^{\star}$.
\end{lem}

\begin{proof}
Since the sampled set $C$ contains at least one vertex from each core $\core_i^{\alpha}$ with constant probability,  and  only vertices from different cores will remain in $C^{\star}$ by \cref{lem:sameCoreSmall} and the algorithm description, the \textsc{SeedAndTrim} procedure returns
$k$ centers with constant probability.

Now we analyze the runtime. The procedure takes $\tilde{O}(n)$ time to compute the norms of $\{x(u)\}_{u\in V[G]}$,  since the embedding has dimension $O(\log^3 n)$ by assumption. It takes $\tilde{O}(k)$ time to sample $\tilde{O}(k)$ vertices, and  trimming the sampling vertices takes $\tilde{O}(k^2)$ time. Hence, the total runtime is  $\tilde{O}(n + k^2)$ .
\end{proof}

As the end of this subsection, we would like to mention that choosing good candidate centers is crucial for
most $k$-means algorithms, and has been studied extensively
in the literature~(e.g.~\cite{k-means++,ORSS12}). Comparing with recent algorithms that obtain good initial centers by
iteratively picking points from a \emph{non-uniform} distribution and take $\Omega(nk)$ time,  our seeding step~(Algorithm~\ref{fig:seed}) runs in $\tilde{O}(n + k^2)$ time.

\subsection{The Grouping Step}
\label{sec:group}

After the seeding step, with constant probability we obtain a set of $k$ vertices $C^{\star}=\{ c_1,\cdots, c_k \}$, and these $k$ vertices belong to $k$ different clusters.   Now we   assign each remaining vertex $u$ to a cluster $S_i$ if, comparing with all other points $x(c_j)$ with $c_j \in C^{\star}$,  $x(u)$ is closer to $x(c_i)$.
A naive implementation of this step requires $\tilde{\Omega}(nk)$ time.
To speed it up, we apply $\eps$-approximate nearest neighbor data structures~($\eps$-\textsf{NNS})~\citep{IndykMotwani98}, whose formal description is as follows:

\begin{prob}[$\eps$-approximate nearest neighbor  problem]
Given a set of point $P \subset \R^d$ and a point $q \in \R^d$, find a point $p \in P$ such that, for all $p' \in P$, $\|p-q\| \le (1+ \eps)\|p'-q\|$.
\end{prob}

\begin{thm}[\cite{IndykMotwani98}]
\label{thm:nns}
Given a set $P$ of points in $\R^d$, there is an algorithm that 
solves the $\eps$-approximate nearest neighbor problem with $$\tilde{O}\left(|P|^{1+\frac{1}{1+\eps}} + d\cdot |P|\right)$$
preprocessing time and $\tilde{O}\left(d\cdot |P|^{\frac{1}{1+\eps}}\right)$ query time.
\end{thm}

Now we set $P=\{x(c_1),\dots,x(c_k)\}$, and apply the above $\varepsilon$-approximate nearest neighbor data structures
to assign the remaining vertices to $k$ clusters $A_1,\cdots, A_k$. By  \cref{thm:nns} and setting $\eps=\log{k} -1$, this step can be finished with $\tilde{O}(k)$ preprocessing time and $\tilde{O}(1)$ query time for each query. Hence, the runtime of the grouping step is $\tilde{O}(n)$. Notice that, with our choice of $\eps=\log{k} -1$ and application of $\eps$-\textsf{NNS}, all the remaining vertices in $V\setminus C^{\star}$ \emph{might not} assign to the cluster $A_i$ with the closest center $c_i$. We will prove in the next subsection that our choice of $\varepsilon$ suffices to obtain a good approximation of the optimal partition.
The runtime of the grouping step, and the properties of the returned clusters are summarized in the following lemma:


\begin{lem}
\label{lem:group}
Given a set of centers $C^{\star}=\{c_1,\dots,c_k\}$, the grouping step runs in $\tilde{O}(n)$ time and returns a partition $A_1,\dots,A_k$ of vertices, such that 
for any $i\in\{1,\ldots, k\}$, and every $u\in A_i$, it holds for any $j\neq i$ that
\[
\|x(u) - x(c_i)\| \le \log k \cdot {\|x(u) - x(c_j)\|}.
\]
\end{lem}
\begin{proof}
The statement follows from the definition of $\varepsilon$-\textsf{NNS} with the choice of $\eps=\log{k} -1$, and \cref{thm:nns}.
\end{proof}

  \subsection{Approximation Analysis of the Algorithm}
  \label{sec:approx}
  Now we study the approximation ratio of the $k$-way partition computed by the seeding and grouping steps. The next lemma analyzes the symmetric difference between the optimal partition and the output of the  algorithm. 

\begin{lem}
\label{lem:approx_group}
Let $A_1,\ldots, A_k$  be the output of the grouping procedure. 
Then, under a proper permutation of the indices, with constant probability for any $1 \le i \le k$ it holds that (i) $
\vol(A_i \triangle S_i) = \tilde{O}\left(k^3/ \Upsilon\right) \vol(S_i)$,
and (ii) $
	\phi_G(A_i) = 1.1\cdot\phi_G(S_i) + \tilde{O}\left( k^3 / \Upsilon\right)$.
\end{lem}


\begin{proof}
We assume that  $c_1,\dots,c_k \in V$ are the centers returned by \textsc{SeedAndTrim}, and $\{x(u)\}_{u\in V[G]}$  is the embedding 
we used in the algorithm. Moreover, $\{x(u)\}_{u\in V[G]}$ satisfies \eq{approx_emb1} and \eq{approx_emb}.
We further assume that these sampled $c_1,\dots,c_k \subseteq\bigcup_{i=1}^k\core_{i}^{\alpha}$. By \cref{lem:success}, this holds with constant probability, and we assume that this event happens in the following analysis.  Then, by the second statement of \cref{lem:sameCoreSmall} it holds for any $i\neq j$ that 
\begin{equation}
\label{eq:distCent}
\|x(c_i)-x(c_j)\|^2 = \Omega\left( \frac{1}{ k \cdot\min\{\vol(S_i),\vol(S_j)\}}\right).
\end{equation}
By \cref{lem:group}, it holds for any $1\leq i\leq k$ that 
\begin{align}
\vol(S_i\setminus A_i)
&\le \sum_{i\ne j} \vol\left( \left\{v \in S_i : \|x(c_i) - x(v)\| >  \frac{\|x(c_j) - x(v)\|}{\log k}\right\}\right)  \nonumber\\
&\le \sum_{i\ne j} \vol\left( \left\{v \in S_i : \|x(c_i) - x(v)\| >  \frac{\|x(c_i) - x(c_j)\|-\|x(c_i) - x(v)\|}{\log k}\right\}\right)  \notag\\
&\le \sum_{i\ne j} \vol\left( \left\{v \in S_i : 2\|x(c_i) - x(v)\| >  \frac{\|x(c_i) - x(c_j)\|}{\log k}\right\}\right)  \notag\\
& = \sum_{i\ne j} \vol\left( \left\{v \in S_i : \|x(c_i) - x(v)\| >  \frac{\|x(c_i) - x(c_j)\|}{2\log k}\right\}\right). \notag
\end{align}
By \eq{approx_emb} and the triangle inequality, we have that
\[
\| x(c_i) - x(v) \| \leq \| F(c_i)- F(v) \| + \frac{1}{n^{2.5}} \leq \left\|F(c_i) -p^{(i)} \right\| + \left\| p^{(i)} - F(v) \right\| + \frac{1}{n^{2.5}},
\]
and hence
\begin{align}
\vol(S_i\setminus A_i)
&\leq   \sum_{i\ne j} \vol\left( \left\{v \in S_i : \left\|F(c_i) -p^{(i)} \right\| + \left\| p^{(i)} - F(v) \right\| + \frac{1}{n^{2.5}} >  \frac{\|x(c_i) - x(c_j)\|}{2\log k}\right\}\right)\notag\\
&\le \sum_{i\ne j} \vol\left( \left\{v \in S_i : \left\|p^{(i)} - F(v)\right\| >  \frac{\|x(c_i)-x(c_j)\|}{2\log k} - \left\|F(c_i)-p^{(i)}\right\| - \frac{1}{n^{2.5}}\right\}\right) \notag \\
&\le \sum_{i\ne j} \vol\left( \left\{v \in S_i : \left\|p^{(i)} - F(v)\right\| >  \frac{\|x(c_i)-x(c_j)\|}{2\log k} -  \sqrt{R_i^{\alpha}} - \frac{1}{n^{2.5}} \right\}\right) \notag \\
&\le \sum_{i\ne j} \vol\left( \left\{v \in S_i :\left\|p^{(i)} - F(v)\right\|^2 = \Omega\left( \frac{1}{ k\log^2k \cdot \min\{\vol(S_i),\vol(S_j)\} } \right)\right\}\right) \notag\\
&=  \tilde{O}\left( k^3/ \Upsilon \right) \vol(S_i), \nonumber
\end{align}
where the last equality follows from \cref{lem:optcost}.
For the same reason, we have 
\begin{align*}
\vol(A_i\setminus S_i) &\le \sum_{i\ne j} \vol\left( \left\{v \in S_j : \|x(c_j) - x(v)\| \ge  \frac{\|x(c_i) - x(v)\|}{\log k}\right\}\right)\\
&=  \tilde{O}\left( k^3/ \Upsilon \right) \vol(S_i),
\end{align*}
and therefore
\[
\vol(S_i\triangle A_i)= \vol(S_i\setminus A_i) + \vol(A_i\setminus S_i)=\tilde{O}\left( k^3/ \Upsilon \right) \vol(S_i).
\]
This yields the first statement of the lemma.

The second statement follows by the same argument  used in proving \thmref{anlykmean}.
\end{proof}

\subsection{Fast computation of the required embedding\label{sec:ht}}
\label{sec:heatker}

So far we  assumed the existence of the embedding $\{x(u)\}_{u\in V[G]}$ satisfying \eq{approx_emb1} and \eq{approx_emb}, and analyzed the performance of the seeding and grouping steps. In this subsection, we will present a nearly-linear time algorithm to compute all the required distances used in the seeding and grouping steps. Our algorithm is based on the so-called \emph{heat kernal} of a graph.

Formally, the heat kernel of $G$ with parameter $t\geq0$ is defined by
\begin{equation}\label{eq:heatdef}
\mat{H}_t\triangleq\mathrm{e}^{-t\calL } =\sum_{i=1}^n \ce^{-t\lambda_i} f_if_i^{\rot}.
\end{equation}
We  view the heat kernel as a geometric embedding from $V[G]$ to $\R^n$ defined by
\begin{equation}\label{eq:htembed}
x_t(u)\triangleq\frac{1}{\sqrt{d_u}}\cdot\left( \ce^{-t\cdot\lambda_1}f_1(u),\cdots, \ce^{-t\cdot\lambda_n}f_n(u)\right),
\end{equation}
 and define the $\ell_2^2$-distance between the points  $x_t(u)$ and $x_t(v)$ by
\begin{equation}\label{eq:htdist}
\eta_t(u,v)\triangleq\|x_t(u)-x_t(v)\|^2.
\end{equation}
The following lemma shows that, when $k=\Omega(\log n)$ and   $\Upsilon=\Omega(k^3)$,
 the values of $\eta_t(u,v)$ for all edges $\{u,v\}\in E[G]$ can be approximately computed in $\tilde{O}(m)$
 time.

\begin{lem}\label{lem:approxEmbedding}
Let $k=\Omega(\log n)$ and $\Upsilon=\Omega(k^3)$. Then,  there is $t=O(\mathrm{poly}(n))$ such that  the embedding $\{x_t(u)\}_{u\in V[G]}$ defined in \eq{htembed} satisfies \eq{approx_emb1} and \eq{approx_emb}. 
Moreover, the values of $\eta_t(u,v)$ for all $\{u,v\}\in E[G]$ can be  approximately computed in $\tilde{O}(m)$ time, such that with high probability the conditions \eq{approx_emb1} and \eq{approx_emb} hold for all edges $\{u,v\}\in E[G]$.
\end{lem}

 Our proof of \cref{lem:approxEmbedding} uses the algorithm for approximating the matrix exponential in \cite{osv12} as a subroutine,
whose performance is summarised in \cref{thm:osv12}.  Recall that any $n\times n$ real and  symmetric matrix $\mat{A}$ is  diagonally dominant (\textsf{SDD}), if $\mat{A}_{ii} \ge \sum_{j\neq i} \left|\mat{A}_{ij}\right|$ for each $i=1,\dots,n$. It is easy to see that 
the Laplacian matrix of any undirected graph is  diagonally dominant.

\begin{thm}[\cite{osv12}]\label{thm:osv12}
Given an $n \times n$ \textsf{SDD} matrix $\mat{A}$ with $m_{\mat{A}}$ nonzero entries, a vector $v$ and a parameter $\delta > 0$, there is an algorithm that can compute a vector $x$ such that $\|\mathrm{e}^{-\mat{A}}v - x\| \le \delta \|v\|$ in time $\tilde{O}((m_{\mat{A}}+n) \log(2+\|\mat{A}\|))$, where the $\tilde{O}(\cdot)$ notation  hides $\poly\log n$ and $\poly\log (1/\delta)$ factors.
\end{thm}

\begin{proof}[Proof of \cref{lem:approxEmbedding}]
By the higher-order Cheeger inequality~\eq{highorder}, we have that 
\[
\Upsilon=\frac{\lambda_{k+1}}{\rho(k)} \leq \frac{2 \lambda_{k+1}}{\lambda_k}.
\]
Since $k=\Omega(\log n)$ and $\Upsilon=\Omega(k^3)$, it holds that $400\cdot\log^2n\leq \lambda_{k+1}/\lambda_k$, and there is $t$ such that
\[
t \in \left(\frac{10\cdot \log n}{ \lambda_{k+1}}, \quad\frac{1}{{20\cdot \lambda_k\cdot \log{n}}}\right).
\]
We first show that  the embedding $\{x_t(u)\}_{u\in V[G]}$ with this  $t$ satisfies   \eq{approx_emb1} and \eq{approx_emb}.

By the definition of $\eta_t(u,v)$, we have that
\begin{align}
\eta_t(u,v) & = \sum_{i=1}^n \mathrm{e}^{-2t \lambda_i}\left(\frac{f_i(u)}{\sqrt{d_u}} - \frac{f_i(v)}{\sqrt{d_v}}\right)^2\notag \\
& = \sum_{i=1}^k \mathrm{e}^{-2t \lambda_i}\left(\frac{f_i(u)}{\sqrt{d_u}} - \frac{f_i(v)}{\sqrt{d_v}}\right)^2 + \sum_{i=k+1}^n \mathrm{e}^{-2t \lambda_i}\left(\frac{f_i(u)}{\sqrt{d_u}} - \frac{f_i(v)}{\sqrt{d_v}}\right)^2.\label{eq:decomp}
\end{align}
Notice that it holds  for $1\le i \le k$ that
\begin{equation}\label{eq:aptlambda1}
1-\frac{1}{10\log n}\leq \ce^{-1/(10\log n)} \leq \ce^{-\lambda_i/(10\lambda_k/\log n)} \leq \ce^{-2t\lambda_i}\leq 1,
\end{equation}
and it holds for $k+1 \le i \le n$ that
\begin{equation}\label{eq:aptlambda2}
	\mathrm{e}^{-2t\cdot \lambda_i} \le \mathrm{e}^{-2 \lambda_i \cdot 10\log{n}/\lambda_{k+1}} \le \mathrm{e}^{-10\log{n} \lambda_{k+1}/\lambda_{k+1}} = \frac{1}{n^{20}}.
\end{equation}
Combining \eq{decomp},
\eq{aptlambda1}, and \eq{aptlambda2}, it holds for any $\{u,v\}\in E[G]$ that 
\[
\left( 1-\frac{1}{10\cdot\log n} \right)\cdot \|F(u)-F(v) \|^2 \leq \eta_t(u,v)\leq \|F(u)- F(v) \|^2 + \frac{1}{n^5},
\]
which proves the first statement.

Now we  show that the distances of $\|x_t(u) -x_t(v)\|$ for all edges $\{u,v\}\in E[G]$ can be approximately computed in nearly-linear time.
For any vertex $u\in V[G]$, we define $\xi_u\in\mathbb{R}^n$, where
$(\xi_u)_v=1/\sqrt{d_u}$ if $v=u$, and $(\xi_u)_v=0$ otherwise.
Combining \eq{heatdef} with \eq{htembed} and \eq{htdist}, we have that
$
\eta_t(u, v) = \left\| \mat{H}_{t} \left( \xi_{u} - \xi_{v} \right) \right\|^2$.
We define $\mat{Z}$ to be the operator of error $\delta$ which corresponds to the algorithm described in \cref{thm:osv12}, and
replacing $\mat{H}_{t}$ with  $\mat{Z}$ we get
\[
\left| \left\| \mat{Z} \left( \xi_{u} - \xi_{v} \right) \right\| - \eta^{1/2}_t(u, v) \right|
\leq \delta \left \|\xi_{u} - \xi_{v} \right\|
\leq \delta,
\]
where the last inequality follows from $d_u, d_v \geq 1$. Hence, it holds that
\begin{equation}\label{eq:approxt1}
\eta^{1/2}_t(u, v) - \delta
\leq \left\| \mat{Z} \left( \xi_{u} - \xi_{v} \right) \right\| \leq
\eta^{1/2}_t(u, v) + \delta.
\end{equation}
By applying the Johnson-Lindenstrauss transform in a way
analogous to the computation of effective resistances~(e.g. \cite{KoutisLP12} and \cite{journals/siamcomp/SpielmanS11}), we obtain 
 an $O(\epsilon^{-2} \cdot\log{n} ) \times n$ Gaussian matrix $\mat{Q}$, such that
with high probability it holds  for all $u, v$ that
\begin{equation}\label{eq:approxt2}
\left( 1 - \epsilon \right) \left\| \mat{Z} \left( \xi_{u} - \xi_{v} \right) \right\|
\leq \left\| \mat{Q} \mat{Z} \left( \xi_{u} - \xi_{v} \right) \right\| \leq
\left( 1 + \epsilon \right) \left\| \mat{Z} \left( \xi_{u} - \xi_{v} \right) \right\|.
\end{equation}
Combining \eq{approxt1} and \eq{approxt2} gives us that
\[
\left( 1 - \epsilon \right) \left( \eta^{1/2}_t(u, v) - \delta \right)
\leq \left\| \mat{Q}\mat{Z} \left( \xi_{u} - \xi_{v} \right) \right\| \leq
\left( 1 + \epsilon \right) \left( \eta^{1/2}_t(u, v) + \delta \right).
\]
Squaring both sides and invoking the inequality
$
 (1 - \epsilon )\alpha^2 - (1 + \epsilon^{-1})b^2
\leq (a + b)^2 \leq (1 + \epsilon )\alpha^2 + (1 + \epsilon^{-1})b^2$
 gives
\[
\left( 1 - 5\epsilon \right) \eta_t(u, v) -  2\delta^2 \epsilon^{-1}
\leq \left\| \mat{Q} \mat{Z} \left( \xi_{u} - \xi_{v} \right) \right\|^2 \leq
\left( 1 + 5\epsilon \right) \eta_t(u, v) + 2\delta^2 \epsilon^{-1}
\]
Scaling $\mat{Q}\mat{Z}$ by a factor of $\left( 1 + 5\epsilon \right)^{-1}$,
and appending an extra entry in each vector to create an additive
distortion of $2 \delta \epsilon^{-1}$ then gives the desired bounds
when $\delta$ is set to $\epsilon n^{-6}$. To satisfy the conditions \eq{approx_emb1} and 
\eq{approx_emb} we just need to set $\eps = O(1/\log{n})$.

To analyze the runtime of computing $\left\| \mat{Q} \mat{Z} \left( \xi_{u} - \xi_{v} \right) \right\|^2$ for all edges $\{u,v\} \in E[G]$,
notice that $\mat{Q}$ has only $O( \log^3{n})$ rows.  We can  then  run the approximate exponential algorithm from~\citep{osv12} $O( \log^3{n})$ times,  where  each time we use a different row of $\mat{Q}$ as input.  Since 
$\|\mathcal{L} \|\leq 2$, by \cref{thm:osv12}  we can compute $\mat{Q}\mat{Z}$ in $\tilde{O}(m)$ time.  Notice that 
$\mat{Q}\mat{Z}\xi_u$ is some column of $\mat{Q}\mat{Z}$ after rescaling, therefore we can compute all the required distances in time $\tilde{O}(m)$.
\end{proof}

We remark that the proof above shows  an interesting property about the embedding~\eq{htembed}, i.e.,  for a large value of $k$ and a certain condition on $\Upsilon$, there is always a $t$ such that the values of $\eta_t(u,v)$ gives a good approximation of $\| F(u)-F(v)\|^2$ for all edges $\{u,v\}\in E[G]$.
A similar intuition which views the heat kernel embedding as a weighted
 combination of multiple eigenvectors was discussed in \cite{osv12}.

\subsection{Proof of  \cref{thm:main_informal} \label{sec:pfthm13}}
We proved in Section~\ref{sec:ht} that if
$k=\Omega(\log n )$ and $\Upsilon=\Omega(k^3)$,
there is a 
\begin{equation}\label{eq:tright}
t\in\left( \frac{10\log n}{\lambda_{k+1}}, \quad\frac{1}{20\cdot\lambda_k\cdot\log n}\right)
\end{equation}
such that $\{x_t(u)\}_{u\in V[G]}$ satisfies the conditions   \eq{approx_emb1} and \eq{approx_emb}. Moreover, the values of $\|x_t(u) -x_t(v) \|$ for $\{u,v \}\in E[G]$ can be approximately computed in nearly-linear time\footnote{\cref{lem:approxEmbedding} shows that both of
the embedding $\{x_t(u)\}_{u\in V[G]}$ and the embedding that the algorithm computes in nearly-linear time 
 satisfy the conditions \eq{approx_emb1} and \eq{approx_emb} with high probability. For the ease of discussion, we use 
  $\{x_t(u)\}_{u\in V[G]}$  to express the embedding that the algorithm actually uses.  
}.
However, it is unclear how to approximate $\lambda_k$,
and without this approximation.
Furthermore, without this approximation of $\lambda_k$,
obtaining the desired embedding $\{x(u)\}_{u\in V[G]}$
becomes highly non-trivial.

To overcome this obstacle, we run the seeding and grouping steps for all possible $t$ of the form $2^i$, where $t\in\mathbb{N}_{\geq 0}$, as it allows us to run the seeding and grouping steps with the right values of $t$ at some point. However,  by \eq{htdist} the distance between any pair of embedded vertices decreases when we increase the value of $t$. Moreover, all these embedded points $\{x_t(u)\}_{u\in V[G]}$ tend to  ``concentrate'' around a single point for an arbitrary large value of $t$. To avoid this situation, for every possible $t$ we compute the value of $\sum_{v\in V[G]} d_v \|x_t(v) \|^2$, and the algorithm only moves to the next iteration
if
\begin{equation}\label{eq:tcondition}
\sum_{v\in V[G]} d_v \|x_t(v) \|^2\geq k\left( 1-\frac{2}{\log n} \right).
\end{equation}
By  \cref{lem:exactProb}, \eq{tcondition} is satisfied for all values of $t$ in the right range \eq{tright}, and the algorithm will not terminate before $t=\lfloor\log n/\lambda_{k+1}\rfloor$.
See Algorithm~\ref{fig:algofix} for the formal description of our final algorithm.

\begin{algorithm}
  \begin{algorithmic}[1]\label{algo:fix}

\STATE \textbf{input}:  the input graph $G$, and the number of clusters $k$

	\STATE Let $t=2$.
	\REPEAT
			\STATE Let $ (c_1, \dots, c_{k})=$ \textsc{SeedAndTrim}$\left(k,\{x_t(u)\}_{u\in V[G]}\right)$. 
			\IF {\textsc{SeedAndTrim} returns exactly $k$ points}
				\STATE Compute a partition $A_1, \dots, A_{k}$ of $V[G]$: for every $v \in V[G]$ assign $v$ to its nearest center $c_i$ using the $\eps$-\textsf{NNS} algorithm with $\eps = \log k-1$.
			\ENDIF
	\STATE Let $t=2t$
	\UNTIL {$t > n^{10}$ \textbf{or} $\sum_{v\in V[G]} d_v \|x_t\|^2 < k\left(1 -\frac{2}{ \log{n}}\right)$.}
	\RETURN $(A_1,\cdots,A_k).$
\end{algorithmic}

\caption{A nearly-linear time graph clustering algorithm, $k=\Omega(\log n )$}
\label{fig:algofix}

\end{algorithm}

\begin{lem}\label{lem:combined}
Let $t = \Omega\left(1/(\lambda_k  \cdot \log{n})\right)$, and $t$ satisfies \eq{tcondition}. Suppose that  \textsc{SeedAndTrim} uses the embedding $\{x_t(u) \}_{u\in V[G]}$
 and returns $k$ centers $c_1,\dots,c_k$.  Then, with constant probability,
 the following statements hold:
 \begin{enumerate}
 \item   It holds that 
 \[
\{c_1,\ldots, c_k \}\subseteq\bigcup_{i=1}^k \core_i^{\alpha}.
\]
\item These  $k$ centers belong to different cores, and it holds for any different $i,j$ that 
\[
\| x_t(c_i) -x_t(c_j) \|^2 =\tilde{ \Omega}\left(\frac{1}{k\cdot\vol(S_i)} \right).
\]
\item For any $i=1,\dots,k$, it holds that  \[\sum_{i=1}^k \sum_{u \in S_i} d_u\cdot \|x(u) - x(c_i)\|^2 = \tilde{O}\left(\frac{k^3}{\Upsilon}\right).\]
\end{enumerate}
\end{lem}
\begin{proof}
Since $\|x_t(u) \|$ is decreasing with respect to the value of $t$ for any vertex $u$, by \lemref{exactProb}
for any $t = \Omega\left(1/(\lambda_k  \cdot \log{n})\right)$
we have:
 \[
\sum_{i=1}^k \sum_{u \notin \core_{i}^{\alpha} } {d_u\cdot \|x_t(u)\|^2} \le \sum_{i=1}^k \sum_{u \notin \core_{i}^{\alpha} } {d_u\cdot \left(\|F(u)\|^2 + \frac{1}{n^5}\right)} \leq    \frac{k}{100 K} + \frac{kn^2}{n^5} \le \frac{1}{\log{k}}. 
\]
On the other hand, we only consider  values of $t$ satisfying \eq{tcondition}.
Since every vertex $u$ is sampled with probability proportional to $d_u\cdot \|x_t(u) \|^2$, with constant probability it holds that 
\[
\{c_1,\ldots, c_k \}\subseteq\bigcup_{i=1}^k \core_i^{\alpha},
\]
which proves the first statement.

 Now we prove that these $k$ centers belong to different cores. 
We fix an index $i$,  and assume that $c_i\in S_i$. We will prove that 
\begin{equation}\label{eq:bignormxt}
\| x_t(c_i)\|^2 =\tilde{\Omega}\left(\frac{1}{\vol(S_i)}\right).
\end{equation}
Assume  by contradiction  that \eq{bignormxt} does not hold, i.e., 
\[
\| x_t(c_i)\|^2 =o\left(\frac{1}{\log^c k\vol(S_i)}\right)
\]
for any constant $c$. Then, we have that 
\begin{align*}
\sum_{u\in\core_i^{\alpha} }d_u\cdot \|x_t(u) \|^2 & \leq \sum_{u\in\core_i^{\alpha} } d_u\cdot \left( \| x_t(c_i)\| + \sqrt{R_i^{\alpha}}  \right)^2 \\
& \leq 2\cdot  \sum_{u\in\core_i^{\alpha} } \left(d_u\cdot  \| x_t(c_i)\|^2 +  d_u\cdot R_i^{\alpha} \right) \\
& = o \left(\frac{1}{\log^ck}\right) + o\left(\frac{1}{k^2} \right)\\
& = o \left(\frac{1}{\log^ck}\right).
\end{align*}
Combining this with \eq{tcondition},  the probability that vertices get sampled from $\core_i^{\alpha}$ is 
\[
\frac{\sum_{u\in\core_i^{\alpha} }d_u\cdot \|x_t(u) \|^2}{\sum_{v\in V[G]} d_v \|x_t(v) \|^2}=o\left(\frac{1}{k\cdot \log^ck} \right).
\]
This means if we sample $K=\Theta(k\log k)$ vertices,
vertices in $\core_i^{\alpha}$ will not get sampled
with probability at least $1-1/\log^5 k$.
This contradicts the fact that $c_i\in\core_i^{\alpha}$.
Therefore \eq{bignormxt} holds.

Now, by description of Algorithm~\ref{fig:seed}, we have
for any $j\neq i$:
\[
\| x_t(c_i) -x_t(c_j) \|^2 \geq \frac{\|x(c_i)\|^2}{2\cdot 10^4\cdot k} =\tilde{ \Omega}\left(\frac{1}{k\cdot\vol(S_i)} \right),
\]
where the last equality follows from  \eq{bignormxt}. Since any vertex in $\core_{i}^{\alpha}$ has distance at most $R_i^{\alpha}$ from $c_i$, $c_j$ and $c_i$ belong to different cores. Therefore, the second statement holds.

Finally we turn our attention to the third statement.
 We showed in \cref{lem:approxEmbedding} that, when $t = \Theta\left(1/(\lambda_k  \cdot \log{n})\right)$, the embedding $\{x_t(u)\}_{u\in V[G]}$ satisfies the conditions \eq{approx_emb1} and \eq{approx_emb}.  Hence,  it holds that
\begin{align}
\sum_{i=1}^k \sum_{u \in S_i} d_u\cdot \|x(u)-x(c_i) \|^2
& \leq  \sum_{i=1}^k  \sum_{u \in S_i} \left( d_u\cdot \|F(u)-F(c_i) \|^2 + \frac{1}{n^5} \right) \nonumber\\
& \leq  \sum_{i=1}^k  \sum_{u \in S_i} \left( d_u\cdot \left( \left\|F(u)-p_i \right\|  + \left\|F(c_i) - p_i \right\| \right)^2 + \frac{1}{n^5}\right) \nonumber\\
& \leq \sum_{i=1}^k  \sum_{u \in S_i}  \left(2\cdot d_u\cdot \left( \left\|F(u)-p_i \right\|^2  + \left\|F(c_i) - p_i \right\|^2 \right) + \frac{1}{n^5} \right)\label{eq:boundterm1} 
\end{align}
Notice that by \cref{lem:optcost} we have
\begin{equation}\label{eq:boundterm2}
\sum_{i=1}^k  \sum_{u \in S_i}  d_u\cdot \|F(u)-p_i \| ^2 \leq 1.1k^2/\Upsilon.
\end{equation}
On the other hand, we have $\|F(c_i) -p_i\|^2\leq R_i^{\alpha}$ as $c_i\in\core_i^{\alpha}$, and 
\begin{equation}\label{eq:boundterm3}
\sum_{i=1}^k  \sum_{u \in S_i}  2\cdot d_u\cdot  \left\|F(c_i) - p_i \right\|^2 \leq  \sum_{i=1}^k 2\vol(S_i) \cdot \frac{\alpha\cdot \mathcal{E}_i}{\vol(S_i)}
= \sum_{i=1}^k 2\alpha\cdot\mathcal{E}_i =\tilde{O}\left(\frac{k^3}{\Upsilon}\right).
\end{equation}
Combining \eq{boundterm1} with \eq{boundterm2} and \eq{boundterm3}, we have that 
\[
\sum_{i=1}^k \sum_{u \in S_i} d_u\cdot \|x(u)-x(c_i) \|^2 \leq \tilde{O}\left(\frac{k^3}{\Upsilon}\right) +\sum_{u\in V[G]} \frac{d_u}{n^5} = \tilde{O}\left(\frac{k^3}{\Upsilon}\right).
\]
Moreover, by \eq{htembed} and \eq{htdist} it is straightforward to see that the distance between any embedded vertices decreases as we increase the value of $t$. Hence, the statement holds for any  $t = \Omega\left(1/(\lambda_k  \cdot \log{n})\right)$.
\end{proof}

\begin{lem}\label{lem:aptnew}
Let $A_1,\ldots, A_k$ be a $k$-way partition returned by Algorithm~\ref{fig:algofix}. Then, under a proper permutation of the indices, with constant probability for any $1 \le i \le k$ it holds that (i) $
\vol(A_i \triangle S_i) = \tilde{O}\left(k^4/ \Upsilon\right) \vol(S_i)$,
and (ii) $
	\phi_G(A_i) = 1.1\cdot\phi_G(S_i) + \tilde{O}\left( k^4 / \Upsilon\right)$.\end{lem}

\begin{proof}
We assume that  $c_1,\dots,c_k$ are the centers returned by \textsc{SeedAndTrim} when obtaining $A_1,\ldots, A_k$. By \cref{lem:combined},  with constant probability it holds that $\{ c_1,\dots,c_k\} \subseteq \bigcup_{i=1}^k\core_i^{\alpha}$, and $c_i$ and $c_j$ belong to different cores for $i\neq j$. Without loss of generality, we assume that $c_i\in\core^{\alpha}_{i}$.
 Then,  it holds  that
 \begin{align}
\vol(S_i\setminus A_i)
&\le \sum_{i\ne j} \vol\left( \left\{v \in S_i : \|x(c_i) - x(v)\| \ge  \frac{\|x(c_j) - x(v)\|}{\log k}\right\}\right)  \notag\\
&\le \sum_{i\ne j} \vol\left( \left\{v \in S_i : \|x(c_i) - x(v)\| \ge  \frac{\|x(c_i) - x(c_j)\|-\|x(c_i) - x(v)\|}{\log k}\right\}\right)  \notag\\
&\le \sum_{i\ne j} \vol\left( \left\{v \in S_i : 2\|x(c_i) - x(v)\| \ge  \frac{\|x(c_i) - x(c_j)\|}{\log k}\right\}\right)  \notag\\
&\le \sum_{i\ne j} \vol\left( \left\{v \in S_i : \left\|x(c_i) - x(v)\right\|^2 = \tilde{\Omega}\left( \frac{1}{ k \min\{\vol(S_j),\vol(S_i)\}}\right) \right\}\right) \label{eq:volboundf1} \\
&=  \tilde{O}\left( {k^4}/{ \Upsilon} \right) \vol(S_i),  \label{eq:volbound1f1}
\end{align}
where \eq{volboundf1} follows from the second statement of \cref{lem:combined}.

Similarly, we also have  that
\begin{align*}
\vol(A_i\setminus S_i) &\le \sum_{i\ne j} \vol\left( \left\{v \in S_j : \|x(c_j) - x(v)\| \ge  \frac{\|x(c_i) - x(v)\|}{\log k}\right\}\right)\\
&=  \tilde{O}\left( {k^4}/{\Upsilon} \right) \vol(S_i).
\end{align*}
This yields the first statement of the lemma.
The second statement follows by the same argument used in proving \thmref{anlykmean}.
\end{proof}

\begin{proof}[Proof of  \cref{thm:main_informal}] The approximation guarantee of the returned partition is shown in \cref{lem:aptnew}. For the runtime, notices that we enumerate  at most $O(\poly\log n)$ possible values of $t$.
Furthermore, and for every such possible value of $t$
the algorithm runs in $\tilde{O}(m)$ time.
This includes computing the distances of embedded points
and the seeding / grouping steps.
Hence, the total runtime is $\tilde{O}(m)$.
\end{proof}

\section*{Acknowledgements}

Part of this work was done while He Sun and Luca Zanetti were
at the Max Planck Institute for Informatics, and while
He Sun was visiting the Simons Institute for the Theory of Computing at UC Berkeley.
We are grateful to Luca Trevisan for insightful comments on
an earlier version of our paper,
and to Gary Miller for very helpful discussions about heat kernels on graphs. We also would like to thank Pavel Kolev, and Kurt Mehlhorn~\cite{KM15} for pointing out an omission in an early version of \lemref{approx_clusters2}. This omission was fixed locally without effecting the statement of the main results.

\bibliographystyle{plain}
\bibliography{reference}

\end{document}